\documentclass[12pt, reqno]{amsart}
\usepackage{amsfonts,latexsym,enumerate}
\usepackage{amsmath}
\usepackage{amscd}
\usepackage{float,amsmath,amssymb,mathrsfs,bm,multirow,graphics}
\usepackage[dvips]{graphicx}
\usepackage[percent]{overpic}
\usepackage{amsaddr}
\usepackage[numbers,sort&compress]{natbib}

\usepackage[]{enumerate}		

\usepackage{xcolor} 

\allowdisplaybreaks

 \usepackage[raiselinks=true,%
             bookmarks=true,%
             bookmarksopenlevel=1,%
             bookmarksopen=true,%
             bookmarksnumbered=true,%
             hyperindex=true,%
             plainpages=false,
             pdfpagelabels=true,
             pdfborder={0 0 0}]{hyperref} 
 \usepackage[pdftex]{thumbpdf}

\newcommand{\R}{{\mathbb R}}

\newcommand{\C}{{\mathbb C}}
\newcommand{\D}{{\mathbb D}}

\newcommand{\re}{\text{\upshape Re\,}}
\newcommand{\im}{\text{\upshape Im\,}}

\newtheorem{theorem}{Theorem}[section]

\newtheorem{lemma}[theorem]{Lemma}

\newtheorem{assumption}[theorem]{Assumption}
\newtheorem{remark}[theorem]{Remark}

\newtheorem{figuretext}{Figure}

\numberwithin{equation}{section}

\input epsf
\date{\today}
\title[Construction of solutions of the defocusing NLS equation]
{Construction of solutions of the defocusing nonlinear Schr\"odinger equation with asymptotically time-periodic boundary values}

\author{Samuel Fromm}
\address{Department of Mathematics, KTH Royal Institute of Technology, \\ 100 44 Stockholm, Sweden.}
\email{samfro@kth.se}

\begin{document}

\begin{abstract}
\noindent
We study the defocusing nonlinear Schr\"odinger equation in the quarter plane with asymptotically periodic boundary values.
By studying  an associated Riemann--Hilbert problem and employing nonlinear steepest descent arguments, we construct solutions in a sector close to the boundary whose leading behaviour is described by a single exponential plane wave.
Furthermore, we compute the subleading terms in the long time asymptotics of the constructed solutions. 
\end{abstract}

\maketitle

\noindent
{\small{\sc AMS Subject Classification (2010)}: 35Q55, 41A60, 35Q15.}

\noindent
{\small{\sc Keywords}: Nonlinear Schr\"odinger equation, Riemann--Hilbert problem, asymptotics, initial-boundary value problem, nonlinear steepest descent.}

\setcounter{tocdepth}{1}
\tableofcontents


\section{Introduction}

An effective way of computing the long time asymptotics of solutions of integrable nonlinear PDEs is given by analyzing a related Riemann--Hilbert (RH) problem with the help of the nonlinear steepest descent method introduced by Deift and Zhou in \cite{DeiftZhou}.
In this way the long time asymptotics of the solutions for initial boundary value problems on the half-line have been derived for several equations, provided that the Dirichlet and Neumann boundary values decay as \(t\to\infty\) (see for example \cite{monvel_fokas_shepelsky_2004} and \cite{FokasItsNLS_half-line,FokasItsSung_NLS_half-line} for the KdV and NLS equation, respectively).

In the case of non-decaying boundary values with respect to \(t\to\infty\), much less is known.
One reason for this is that in this case the Riemann--Hilbert approach is generally harder to apply, because not all boundary data necessary for a well-posed problem are known \cite{FokasUnifiedApproach}.

A special case of non-decaying boundary data, especially important from the standpoint applications, is given by asymptotically time-periodic boundary conditions \cite{BonaFokas2008,MKSZ2010,LenellsFokasUnifiedMethodII}. 

In this paper we consider boundary data whose leading order long time behaviour is described by a single exponential.
More precisely, we consider the defocusing nonlinear Schr\"odinger (NLS) equation
\begin{align}\label{nls}
i u_t +u_{xx}-2|u|^2 u=0
\end{align}
in the quarter plane \(\{(x,t)\in\R^2\,|\,x\geq0, t\geq0\}\) with boundary values of the form
\begin{align}\label{nls-boundary-values}
u(0, t) \sim \alpha e^{i\omega t}, \qquad u_x(0,t) \sim ce^{i\omega t}, \qquad t\to \infty,
\end{align}
where \(\alpha>0\), \(\omega\in\R\), and \(c\in\C\) are three parameters.

Boundary values of the type \eqref{nls-boundary-values} for the \textit{focusing} NLS equation 
\begin{align}\label{focnls}
i u_t +u_{xx}+2|u|^2 u=0
\end{align}
have been studied by Boutet de Monvel and coauthors \cite{BIK2007,BIK2009,BKS2009,MKSZ2010,BK2007}.
They showed that there exists a solution of the focusing NLS equation
\eqref{focnls}
in the quarter plane
with boundary values of the form \eqref{nls-boundary-values} and with decay as \(x\to\infty\) if and only if the parameters \((\alpha, \omega,c)\) satisfy either 
 \begin{align}\label{firsttriple}
c=\pm\alpha\sqrt{\omega-\alpha^2}\quad\text{and}\quad \omega\geq \alpha^2
\end{align}
or
\begin{align}\label{analogoustriple}
c=i\alpha\sqrt{|\omega|+2\alpha^2}\quad\text{and}\quad \omega\leq-6\alpha^2.
\end{align}
Using the Deift--Zhou nonlinear steepest descent method they were also able to determine the leading order asymptotics of any such solution.
In particular, in \cite{BIK2009} it was shown that in the case \(\omega< -6\alpha^2\), the quarter plane \(\{x>0,t>0\}\) is divided into three asymptotic sectors:
A plane wave sector \(\{0\leq x/(4t)<\tilde\beta-\alpha\sqrt{2}\}\) (where \(\tilde\beta=\sqrt{{\alpha^2}/2-{\omega}/4}\)), an elliptic wave sector \(\{\tilde\beta-\alpha\sqrt{2}<x/(4t)<\tilde\beta\}\), and a Zakharov-Manakov sector \(\{\tilde\beta<x/t\}\).
In the plane wave sector, the solution \(u\) approaches the plane wave
\begin{align*}
(x,t)\mapsto\alpha e^{2i \tilde\beta x+i \omega t}.
\end{align*}
This behaviour is expected since the sector lies near the boundary. In the elliptic wave and Zakharov-Manakov sectors the asymptotics are described by a modulated elliptic wave and Zakharov-Manakov type formulas, respectively.

The defocusing case has been studied by Lenells \cite{Ldefocusing-admissible} and Fokas and Lenells  \cite{tperiodicI,tperiodicII}.
As in the focusing case, in \cite{Ldefocusing-admissible} \textit{necessary} conditions on the parameter triple \((\alpha, \omega, c)\) for the existence of a solution of \eqref{nls} with boundary values of the type \eqref{nls-boundary-values} were derived. 
However, in the defocusing case this leads to five different families of triples.
Two of these families are analogous to the two sets of triples \eqref{firsttriple} and \eqref{analogoustriple}.
The family corresponding to \eqref{analogoustriple} is given by the family (cf. (2.4) of \cite{Ldefocusing-admissible})
\begin{align}\label{parameterfamilyalt}
&\bigg\{\bigg(\alpha, \omega, c = i\alpha \sqrt{-2\alpha^2-\omega} \bigg)  
\; \bigg|  \alpha>0, \; \omega<-3\alpha^2\bigg\}.
\end{align} 

To motivate the results of this paper, let us recall that the unified transform method, also known as the Fokas method, expresses the solution of the initial boundary value problem for the defocusing NLS equation \eqref{nls} on the half-line with vanishing boundary condition at $x = \infty$ in terms of the solution of a \(2\times2\)-matrix valued RH problem \cite{FokasUniformTransform,FokasUnifiedApproach,FokasItsSung_NLS_half-line}.
This RH problem is formulated in terms of spectral functions \(r_1(k)\) and \(h(k)\), which in turn are defined in terms of the initial and boundary values.
One can also construct solutions of \eqref{nls} by solving the RH problem for independent spectral data \(r_1(k)\) and \(h(k)\), assuming that these functions satisfy certain conditions.
Thus, a natural question to ask in this situation is: What conditions need to be imposed on \(r_1(k)\) and \(h(k)\) in order for the corresponding RH problem to give rise to a solution of \eqref{nls} on the half-line with asymptotically time-periodic boundary values? 
This question is also of interest with respect to the results presented in \cite{Ldefocusing-admissible} as outlined above: Are the conditions on the parameter triple \((\alpha, \omega, c)\) given by \eqref{parameterfamilyalt} also \textit{sufficient} for the existence of a solution of \eqref{nls}  with decay as \(x\to\infty\) and with boundary values satisfying \eqref{nls-boundary-values}?
The aim of the work at hand is to take first steps towards answering these questions.

We will show that for any triple from the family \eqref{parameterfamilyalt}, we may construct solutions of \eqref{nls} with boundary values satisfying \eqref{nls-boundary-values}, at least in a sector close to the boundary .
More precisely, we will provide a class of independent spectral data \(r_1(k)\) and \(h(k)\) which gives rise to solutions of \eqref{nls} in the plane wave sector \(\{0\leq\frac{x}t<4\beta-2\alpha-\delta\}\) for \(\delta>0\) small and \(t\) large, and which have boundary values of the form \eqref{nls-boundary-values}.
Furthermore, we will compute the first two terms in the long time asymptotics of the constructed solutions and their \(x\)-derivatives.

The proof is based on the Riemann--Hilbert approach and combines the uniform transform method for nonlinear integrable PDEs introduced by Fokas \cite{FokasUniformTransform} with the nonlinear steepest descent method introduced by Deift and Zhou \cite{DeiftZhou}.
First, we will transform the associated RH problem, containing the spectral data \(r_1(k)\) and \(h(k)\), into a small norm RH problem, which can be solved for large \(t\) in the plane wave sector via a Neumann series.
Using ideas based on the dressing method \cite{ZakharovShabat1,ZakharovShabat2}, it can be shown that the solutions of the original RH problem leads to a solution \(u\) of \eqref{nls}.
Finally, by studying the asymptotics of the solution of the associated RH problem, we compute the asymptotics of the solution \(u\) and its \(x\)-derivative \(u_x\).

A crucial step in the proof is the construction of analytic approximations of the spectral functions defined on the jump contour. 
Our construction of such approximations is based on the original approach of \cite{DeiftZhou}; however, since one of the spectral functions has singularities on the jump contour, the construction of an analytic approximation of this function requires some novel ideas.

%

\subsection*{Acknowledgement}
The author thanks Jonatan Lenells for helpful discussions.
Support  is  acknowledged  from  the  European  Research  Council,  Grant  Agreement No. 682537.

\section{Notation}

The family \eqref{parameterfamilyalt} can be written as (let \(K\to\beta\) in Section 5.3 in \cite{Ldefocusing-admissible})
\begin{align}\label{parameterfamily}
&\bigg\{\bigg(\alpha  = \sqrt{\frac{|\omega|}{2} - 2\beta^2}, \omega, c = 2i\alpha \beta \bigg)  
\; \bigg|  -12 \beta^2 < \omega < -4\beta^2, \; \beta > 0\bigg\}.
\end{align} 
For a given parameter triple \((\alpha, \omega, c)\) belonging to the family
\eqref{parameterfamily} and the associated parameter \(\beta=\sqrt{{|\omega|}/4-{\alpha^2}/2}\), define a function \(\Omega(k)\) by
$$\Omega(k) = 2(k-\beta)X(k),\quad k\in\C\setminus[E_1,E_2],$$
where 
$$X(k) = \sqrt{(k-E_1)(k - E_2)},\qquad E_1 = -\beta - \alpha, \qquad E_2 = -\beta + \alpha,$$
with a branch cut along \([E_1,E_2]\) and the branch of the square root being choosen such that
$$X(k)=\sqrt{(k+\beta)^2-\alpha^2}=k+\beta+O(k^{-1}),\qquad k\to\infty.$$
Similarly, define a function \(\Delta(k)\) by
\begin{align}\label{defindelta}
\Delta(k) = \bigg(\frac{k-E_2}{k-E_1}\bigg)^{1/4},\quad k\in\C\backslash[E_1,E_2],
\end{align} 
with the branch being chosen such that \(\Delta(k)=1+O(k^{-1})\) as \(k\to\infty\).
We orient the interval \([E_1,E_2]\), viewed as a contour in the complex plane, from left to right and denote by \(\Delta_+(s)\) and \(\Delta_-(s)\) for \(s\in(E_1,E_2)\) the boundary values of the function \(\Delta\) from the left and right, respectively.
Similar notation is used for other functions.

\begin{figure}
\begin{center}
 \begin{overpic}[width=.55\textwidth]{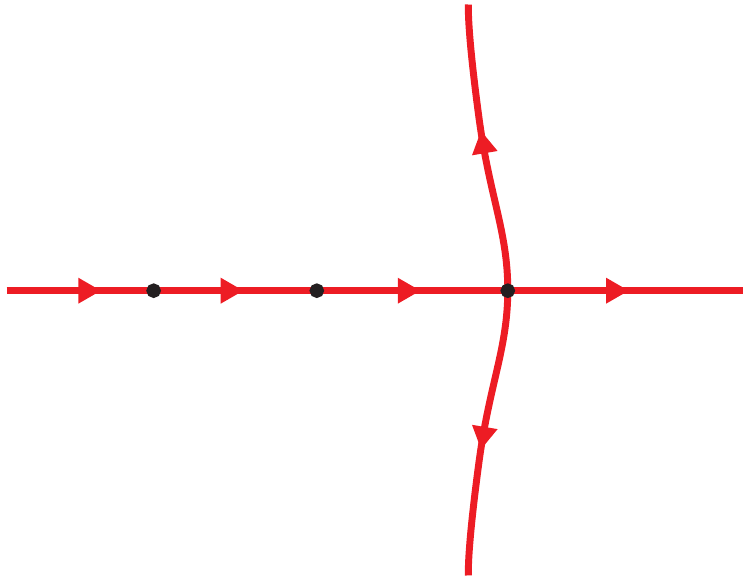}
      \put(102,37.5){\small $\Gamma$}
      \put(18,33){\small $E_1$}
      \put(40,33){\small $E_2$}
      \put(69,35){\small $\kappa_+$}
      \put(81,34){\small $1$}
      \put(60,56){\small $2$}
      \put(11,34){\small $3$}
      \put(54,34){\small $3$}
      \put(60,18){\small $4$}
      \put(30,34){\small $5$}
      \put(80,60){\small $D_1$}
	\put(30,60){\small $D_2$}
	\put(30,10){\small $D_3$}
	\put(80,10){\small $D_4$}
    \end{overpic}
     \begin{figuretext}\label{Gamma.pdf}
        The contour  $\Gamma$ and the domains \(D_j\), \(j=1,2,3,4\), in the complex $k$-plane. 
         \end{figuretext}
     \end{center}
\end{figure}

Define domains \(D_j\subseteq \C\), \(j=1,2,3,4\), by (see Figure \ref{Gamma.pdf})
\begin{align*}
D_1&=\{k\in\C\,|\, \im k>0,\im \Omega(k)>0\},
&D_2=\{k\in\C\,|\, \im k>0,\im \Omega(k)<0\},\\
D_3&=\{k\in\C\,|\, \im k<0,\im \Omega(k)>0\},
&D_4=\{k\in\C\,|\, \im k<0,\im \Omega(k)<0\},
\end{align*}
and let $\Gamma = \R \cup (\bar{{D}}_1 \cap \bar{{D}}_2) \cup (\bar{{D}}_3 \cap \bar{{D}}_4)$ 
denote the contour separating the domains ${D}_j$ oriented as in Figure \ref{Gamma.pdf}.
For complex functions \(r(k)\), \(r_1(k)\) and \(h(k)\), defined on appropriate parts of \(\Gamma\), and \(x,t>0\) define a jump matrix \(v(x,t,k)\) for \(k\in\Gamma\) by
\begin{align}\label{jumpmatrixRHproblem}
v(x,t,k)=
\begin{cases}
\begin{pmatrix}1 - |r_1(k)|^2 & \overline{r_1(k)} e^{-2i(kx+2k^2t) } \\
-r_1(k) e^{2i(kx+2k^2 t)} & 1 \end{pmatrix}, &k\in\bar{D}_1\cap\bar{D}_4,
\\
\begin{pmatrix}1 & 0 \\
-h(k) e^{2i(kx+2k^2 t)} & 1 \end{pmatrix},&k\in\bar{D}_1\cap\bar{D}_2,
\\
 \begin{pmatrix}1 - |r(k)|^2 & \overline{r(k)} e^{-2i(kx+2k^2 t)} \\
-r(k) e^{2i(kx+2k^2 t)} & 1 \end{pmatrix},&k\in\bar{D}_2\cap\bar{D}_3,
\\
\begin{pmatrix}1 & \overline{h(\bar{k})} e^{-2i(kx+2k^2 t)} \\
0 & 1 \end{pmatrix},&k\in\bar{D}_3\cap\bar{D}_4.
\end{cases}
\end{align}
%
%
Put \(\zeta\!:=x/t\). For \(0\leq \zeta<4\beta-2\alpha\) define real numbers \(\kappa_+\) and \(k_0\) by
\begin{align}\label{definition_points_of_interest}
\begin{split}
 &k_0\equiv k_0(\zeta)=-\frac{4\beta + \zeta}{8} + \sqrt{\frac{\alpha^2}{2} + \Big(\frac{4\beta - \zeta}{8}\Big)^2},
 \\
 &\kappa_{+}=\bigcap_{j=1}^4 \bar{D}_j=k_0(0)=-\frac{\beta}{2} + \sqrt{\frac{\alpha^2}{2} + \frac{\beta^2}{4}}.
 \end{split}
\end{align}

Finally, let us introduce some general notation, not specific to our problem.
Given an open connected set \(D\subseteq \C\) bounded by a piecewise smooth curve \(\partial D\subseteq \C\cup\{\infty\}\), we say that an analytic function \(f\colon D\to \C\) belongs to the Smirnoff class  \(\dot{E}^2(D)\) if there exist piecewise smooth curves \(\{C_n\}_{n=1}^\infty\) in \(D\) tending to \(\partial D\) in the sense that \(C_n\) eventually surrounds each compact subdomain of \(D\) and such that
\begin{align*}
\sup_{n \geq 1} \int_{C_n} |f(z)|^2 |dz| < \infty.
\end{align*}
In the case that \(D=D_1 \cup \cdots \cup D_m\) is a finite union of such open subsets, \(\dot{E}^2(D)\) denotes the space of all functions \(f:D\to \C\) satisfying \(f|_{D_j} \in \dot{E}^2(D_j)\) for each \(j\).
The space of bounded analytic functions on \(D\) is denoted by \(E^\infty(D)\).

For a piecewise smooth contour \(\Sigma\) in \(\C\cup\{\infty\}\) and real numbers \(a<b\), we say that the sector \(W_{a,b}=\{k\in\C\,|\, a\leq \arg k\leq b\}\) defines a nontangential sector at \(\infty\) with respect to the contour \(\Sigma\), if there exists a \(\delta>0\) such that \(\{a-\delta\leq \arg k\leq b+\delta\}\) does not intersect \(\Sigma\cap\{|k|>R\}\) for \(R>0\) large enough.
We say that a function \(f\) of \(k\in\C\setminus\Sigma\) has nontangential limit \(L\) at \(\infty\), denoted by
$$
\overset{\angle}{\lim_{k \to \infty}} f(k)=L,
$$
if \(\lim_{k\to\infty, k\in W_{a,b}}f(k)=L\) for every nontangential sector \(W_{a,b}\) at \(\infty\).

Lastly, we let \(C>0\) denote a generic constant which may change within a computation.

\section{Main Result}

Before stating our main result we introduce the assumptions under which the result is valid.
For reasons of clarity, the assumptions are split up into two parts.
The first part of our assumptions is somewhat standard in this context, in the sense that spectral data originating from a solution of \eqref{nls} typically satisfies these assumptions if the initial and boundary data have sufficient smoothness and decay and are compatible  \cite{FokasUnifiedApproach,HuangLenellsSineGordonquarter,FokasItsSung_NLS_half-line,mKdVLenells}. 
In particular, the so called global relation typically implies part (\ref{globalrelationassumption}) of the assumptions below.

\begin{assumption}\label{assumptions-spectral-functions}
Suppose that \(r_1\colon\R\to \C\) and \(h\colon\bar{D}_2\to \C\) are continuous functions satisfying the following properties:
\begin{enumerate}[(a)]
\item \(h\) is analytic in \(D_2\) and \(h\in C^6(\bar{D}_1\cap \bar{D}_2)\).
\item \(r_1\in C^6((E_2,\infty))\).
\item The continuous function \(r\colon(-\infty,\kappa_+]\to\C\) defined by \(r(k)=r_1(k)+h(k)\) is \(C^3\) on \((-\infty,E_1)\cup(E_1,E_2)\) and \(C^6\) on \((E_2,\kappa_+]\).
\item \(|r|<1\) on \((-\infty,\kappa_+]\setminus[E_1,E_2]\).
\item \label{globalrelationassumption} It holds that \(r(\kappa_+)=0\). 
\item \label{seriesexpansionsassumptions} There exist constants \(h_j\in\C\), \(j=1,2,3,4\), such that
 \begin{align}\label{seriesexpansionh}
\frac{d^n}{dk^n}h(k)=
	\frac{d^n}{dk^n}\left(\sum_{j=1}^4\frac{h_j}{k^j}\right)+O(k^{-5-n}),&\qquad k\to\infty,\, k\in \bar{D}_2,
\end{align}
and
 \begin{align}\label{seriesexpansionr1}
\frac{d^n}{dk^n}r_1(k)=
	\frac{d^n}{dk^n}\left(\sum_{j=1}^4\frac{-h_j}{k^j}\right)+O(k^{-5-n}),&\qquad |k|\to\infty,\, k\in\R,
\end{align}
for \(n=0,1,2\).
\end{enumerate}
\end{assumption}

\begin{remark}
Note that \eqref{seriesexpansionh} and \eqref{seriesexpansionr1} imply
\begin{align}\label{seriesexpansionrinfinity}
\frac{d^n}{dk^n}r(k)=
	O(k^{-5-n}),&\qquad k\to-\infty,\, k\in \R.
\end{align}
\end{remark}

In order to motivate the second part of our assumptions, we recall that our goal is to construct solutions approaching the background plane wave solution $u^b(x,t)$ of \eqref{nls} defined by
\begin{align}\label{backgroundsolution}
u^b(x,t) = \alpha e^{2i\beta x + i\omega t}.
\end{align}
This problem bears similarities to shock problems for the defocusing NLS on the whole line, see for instance \cite{Jenkins2015}.
In view of the spectral data found in \cite{Jenkins2015} and the spectral data appearing for the analogous problem for the focusing NLS \cite{BIK2009} (the latter having a single cut in the complex plane), we choose our data near the branch cut to resemble the boundary values \(r^b_+(k)\) for \(k\in\R\) of the function \(r^b(k)\) defined by
\begin{align}\label{rb}
r^b(k)=i\frac{\Delta(k)-\Delta^{-1}(k)}{\Delta(k)+\Delta^{-1}(k)},\quad k\in\C\setminus[E_1,E_2],
\end{align}
where \(\Delta\) is defined by \eqref{defindelta}.
The function \(r^b\) has a single cut from \(E_1\) to \(E_2\) with \(|r^b_+|=1\) on \([E_1,E_2]\), \(r^b_+(E_j)\in\{\pm i\}\), \(j=1,2\), and \(r^b_+(k)\) behaves like a square root as \(k\) approaches one of the branch points \(E_1\) and \(E_2\).
The assumptions below capture this behaviour.
For a more technical motivation we refer to Remark \ref{derivationofmodeldata}.

\begin{assumption}[Assumptions near the branch cut]\label{assumptions-spectral-functions-branch}
Suppose that the function \(r\colon(-\infty,\kappa_+]\to\C\) defined in Assumption \ref{assumptions-spectral-functions} (c) satisfies:
\begin{enumerate}[(a)]
\item $|r|=1$ on $ [E_1,E_2]$.
\item \(r\neq\pm i\) on \((E_1,E_2)\), \(r(E_1)\in\{\pm i\}\), and \(r(E_2)\in\{\pm i\}\).
\item Near the branch points the function \(r\) admits series expansions of the form
\begin{align}\label{seriesexpansionr}
r(k)=
	\begin{cases}
	i\left(\sum_{l=0}^7q_{2,l}(k-E_2)^{l/2}\right)+O((k-E_2)^{4}),&k\downarrow E_2,
		\\ 
	i\left(\sum_{l=0}^7i^lq_{2,l}(E_2-k)^{l/2}\right)+O((E_2-k)^{4}),&k\uparrow E_2,
		\\
	i\left(\sum_{l=0}^7i^lq_{1,l}(k-E_1)^{l/2}\right)+O((k-E_1)^{4}),&k\downarrow E_1,
		\\ 
	i\left(\sum_{l=0}^7(-1)^{l}q_{1,l}(E_1-k)^{l/2}\right)+O((E_1-k)^{4}),&k\uparrow E_1,
	\end{cases}
\end{align}
where \(q_{i,l}\) are real coefficients with \(q_{i,1}\neq 0\) and the expansions in \eqref{seriesexpansionr} can be differentiated termwise three times.
\end{enumerate}
\end{assumption}

\begin{remark}
We note that Assumption \ref{assumptions-spectral-functions} (d), Assumption \ref{assumptions-spectral-functions-branch} (a) and  Assumption \ref{assumptions-spectral-functions-branch} (b) pose additional conditions on the coefficients \(q_{i,l}\).
In particular, it follows that
\begin{align}\label{cancellation_assumption}
q_{i,0}\in\{\pm1\},\quad (-1)^iq_{i,0}q_{i,1}<0,\quad 2q_{i,0}q_{i,2}=q_{i,1}^2,\quad i=1,2.
\end{align}
\end{remark}

Our main result reads as follows. The sector in question is displayed in Figure \ref{quarter_plane_regions_filled_final.pdf}.

\begin{figure}
\begin{center}
 \begin{overpic}[width=.55\textwidth]{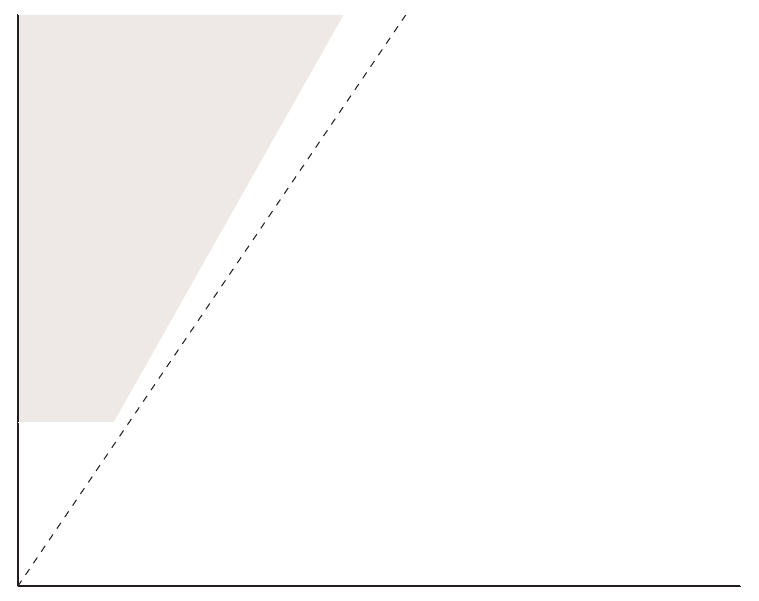}
      \put(-1,-2){\small $0$}
      \put(2,78){\small $t$}
      \put(-2,22){\small $T$}
      \put(14,51){\small $S$}
      \put(99,1){\small $x$}
      \put(54,78){\small $\zeta=4\beta-2\alpha$}
    \end{overpic}
     \begin{figuretext}\label{quarter_plane_regions_filled_final.pdf}
        The sector \(S\) in which Theorem \ref{maintheorem} applies is shaded.
         \end{figuretext}
     \end{center}
\end{figure}

\begin{theorem}\label{maintheorem}
Assume that the parameter triple \((\alpha,\omega,c)\) belongs to the family \eqref{parameterfamily}. Let \(r_1\), \(h\), and \(r=r_1+h\) be functions satisfying Assumptions \ref{assumptions-spectral-functions} and \ref{assumptions-spectral-functions-branch}.
\begin{enumerate}[(a)]
\item For each \(c_0\in(0,4\beta-2\alpha)\) there exists a time \(T>0\) such that the \(L^2\)-RH problem
\begin{align}\label{RHm}
\begin{cases}
m(x, t, \cdot) \in I + \dot{E}^2(\C \setminus \Gamma),\\
m_+(x,t,k) = m_-(x, t, k) v(x, t, k) \quad \text{for a.e.} \ k \in \Gamma,
\end{cases}
\end{align}
where the jump matrix \(v\) is defined by \eqref{jumpmatrixRHproblem}, has a unique solution for each \((x,t)\in S\!:=\{0\leq \frac{x}t\leq c_0\}\cap\{t\geq T\}\).

\item The nontangential limit 
\begin{align}\label{ulim}
u(x,t) = -2i D(0,\infty)^2 \overset{\angle}{\lim_{k \to \infty}} (k\,m(x,t,k))_{12}
\end{align}
exists for each \((x,t)\in S\) and  the function \(u\colon S\to\C\) defined by \eqref{ulim} is \(C^2\) in \(x\) and \(C^1\) in \(t\) and satisfies \eqref{nls} for \((x,t)\in S\).
Here \(D(0,\infty)\) is a constant of absolute value $1$ given by
\begin{align}\label{D0infinitydefinition}
D(0,\infty)= e^{-\frac{1}{2\pi i}\big\{\big(\int_{-\infty}^{E_1} + \int_{E_2}^{\kappa_+}\big) \frac{\ln(1 - |r(s)|^2)}{X(s)} ds
+ \int_{E_1}^{E_2} \frac{i \arg(r(s))}{X_+(s)} ds\big\}}
\end{align}
where \(\arg(r)\) is continuous on \([E_1,E_2]\) and \(\arg(r(E_2))\in(-\pi,\pi]\).

\item The \(x\)-derivative of \(u\) is given by
\begin{align}\label{uxlim}
u_x(x,t)&=-D(0,\infty)^2\overset{\angle}{\lim_{k \to \infty}}\left(4k^2m_{12}(x,t,k)+2i u(x,t) k m_{22}(x,t,k)\right) 
\end{align}
for \((x,t)\in S\).

\item As \(t\to\infty\), the following asymptotic expansions are valid uniformly for  \(\zeta\in[0,c_0]\):
\begin{align}\label{asymptoticexpansionformulas}
\begin{aligned}
u(x,t) &= e^{-\frac{1}{\pi i}\int_{k_0}^{\kappa_+} \frac{\ln(1 - |r(s)|^2)}{X(s)} ds}\alpha e^{2i\beta x +i\omega t}+\frac{u_a(x,t)}{\sqrt{t}}+O\!\left(\frac{\ln t}{t}\right),
	\\
u_x(x,t)&=2ie^{-\frac{1}{\pi i}\int_{k_0}^{\kappa_+} \frac{\ln(1 - |r(s)|^2)}{X(s)} ds}\alpha\beta e^{2i\beta x +i\omega t}+\frac{u_b(x,t)}{\sqrt{t}}+O\!\left(\frac{\ln t}{t}\right).
\end{aligned}
\end{align}
The subleading coefficients \(u_a\) and \(u_b\) are given explicitly by
\begin{align*}
u_a(x,t)
&=e^{2i\beta x +i\omega t}e^{-\frac{1}{\pi i}\int_{k_0}^{\kappa_+} \frac{\ln(1 - |r(s)|^2)}{X(s)} ds}
\Bigg(\frac{it^{-i\nu}{\beta^X_{k_0}}(\Delta(k_0)^2+1)^2}{2e^{2itg(k_0)}\Delta(k_0)^2\psi(k_0)^{1+2i\nu}D_b(k_0)^{-2}}
\\&
\qquad+\frac{it^{i\nu}\overline{\beta^X_{k_0}}(\Delta(k_0)^2-1)^2}{2e^{-2itg(k_0)}\Delta(k_0)^2\psi(k_0)^{1-2i\nu}D_b(k_0)^2}\Bigg)
\end{align*}
and
\begin{align*}
u_b(x,t)
&=e^{2i\beta x +i\omega t}e^{-\frac{1}{\pi i}\int_{k_0}^{\kappa_+} \frac{\ln(1 - |r(s)|^2)}{X(s)} ds}\Bigg(\frac{t^{-i\nu}\beta^X_{k_0}\left(\frac{-2\alpha^2}{k_0-E_1}+k_0(\Delta(k_0)^2+1)^2\right)}{e^{2itg(k_0)}\Delta(k_0)^2\psi(k_0)^{1+2i\nu}D_b(k_0)^{-2}}
\\&\qquad+\frac{t^{i\nu}\overline{\beta^X_{k_0}}\left(\frac{-2\alpha^2}{k_0-E_1}+k_0(\Delta(k_0)^2-1)^2\right)}{e^{-2itg(k_0)}\Delta(k_0)^2\psi(k_0)^{1-2i\nu}D_b(k_0)^2}\Bigg),
\end{align*}
where \(\Delta\) is defined by \eqref{defindelta} and \(g(k_0)\), \(\nu\), \(\beta_{k_0}^X\), \(\psi(k_0)\) and \(D_b(k_0)\) are functions of \(\zeta\in[0,c_0]\) defined by
\begin{align*}
g(k_0)&=(2k_0-2\beta+\zeta)X(k_0),\quad \nu=-\frac1{2\pi}\ln(1-|r(k_0)|^2),
\\
\beta_{k_0}^X &= \sqrt{\nu} e^{i\left(\frac{3\pi}{4} - \arg (-r(k_0)) + \arg \Gamma(i\nu )\right)},\quad \psi(k_0)=\frac{2\sqrt{2}\left({\frac{\alpha^2}{2} + \left(\frac{4\beta - \zeta}{8}\right)^2}\right)^{1/4}}{\sqrt{X(k_0)}},
\\
D_b(k_0)&=\lim_{z\to k_0,z>k_0}\left[(z-k_0)^{-i\nu} e^{\frac{X(z)}{2\pi i}\big\{\big(\int_{-\infty}^{E_1} + \int_{E_2}^{k_0}\big) \frac{\ln(1 - |r(s)|^2)}{X(s)(s - z)} ds
+ \int_{E_1}^{E_2} \frac{i\arg(r(s))}{X_+(s)(s - z)} ds\big\}}\right].
\end{align*}
\end{enumerate}
\end{theorem}

Before proving Theorem \ref{maintheorem} we will make a series of remarks as well as provide a class of spectral functions satisfying our assumptions.

\begin{remark}[Behaviour on the Boundary]
For \(x=0\) we have \(k_0(0)=\kappa_+\). Thus the exponential factor \(e^{-\frac{1}{\pi i}\int_{k_0}^{\kappa_+} \frac{\ln(1 - |r(s)|^2)}{X(s)} ds}\) 
appearing in the leading terms  of the asymptotic expansions of \(u\) and \(u_x\) is equal to \(1\).
Furthermore, Assumption  \ref{assumptions-spectral-functions} \textit{(e)} states that \(r(\kappa_+)=0\) which implies \(\nu=0\) and \(\beta^X_{k_0}=0\) for \(x=0\).
Hence the subleading terms \(u_a\) and \(u_b\) in the asymptotic expansions of \(u\) and \(u_x\) vanish for \(x=0\).
In conclusion, on the boundary \(x=0\), the expansions \eqref{asymptoticexpansionformulas} read as 
$u(0,t) = \alpha e^{i\omega t}+O\!\left(\frac{\ln t}{t}\right)$ and $u_x(0,t)=2i\alpha\beta e^{i\omega t}+O\!\left(\frac{\ln t}{t}\right)$.
Thus, since \(c=2i\alpha\beta\) (cf. \eqref{parameterfamily}), the boundary values of \(u\) indeed satisfy \eqref{nls-boundary-values}. 
\end{remark}

\begin{remark}
Instead of using formula \eqref{uxlim}, the asymptotics of \(u_x\) can alternatively be  computed by differentiating \eqref{ulim}, see Remark \ref{comparisonuxderivation}.
\end{remark}

\begin{remark}
The assumptions on \(r\) on the branch cut imply that we can always choose \(\arg\) as stated in part (b) of Theorem \ref{maintheorem}.
Furthermore, since 
$$
-\frac1{2\pi i} \int_{E_1}^{E_2} \frac{2\pi i }{X_+(s)} ds=i\pi,
$$
a different choice of the branch of \(\arg\) changes the sign of \(D(0,\infty)\) but does not affect the final result.
\end{remark}


\begin{remark}[Derivation of the Model Data]\label{derivationofmodeldata}
The goal of this paper is to construct solutions of \eqref{nls} with boundary values of the form \eqref{nls-boundary-values} and decay as \(x\to\infty\), using the Riemann--Hilbert approach.
To find a candidate for the spectral data let us consider a particular set of initial and boundary data:
Assume that our solution \(u\) coincides with the background solution \(u^b\) (cf. \eqref{backgroundsolution}) on the boundary and let us assume that \(u(x,0)=0\) for \(x> 0\).
The uniform transform method \cite{FokasUniformTransform,FokasUnifiedApproach} then leads to a Riemann--Hilbert problem with jumps along \(\Gamma\) and a jump matrix \(v^b\) given by
\begin{align}\label{jumpmatrixmodelproblem}
v^b(x,t,k)=
\begin{cases}
\begin{pmatrix}1 - |r_1^b(k)|^2 & \overline{r_1^b(k)} e^{-2i(kx+2k^2t) } \\
-r_1^b(k) e^{2i(kx+2k^2 t)} & 1 \end{pmatrix}, &k\in\bar{D}_1\cap\bar{D}_4,
\\
\begin{pmatrix}1 & 0 \\
-h^b(k) e^{2i(kx+2k^2 t)} & 1 \end{pmatrix},&k\in\bar{D}_1\cap\bar{D}_2,
\\
 \begin{pmatrix}1 - |r_+^b(k)|^2 & \overline{r_+^b(k)} e^{-2i(kx+2k^2 t)} \\
-r_+^b(k) e^{2i(kx+2k^2 t)} & 1 \end{pmatrix},&k\in\bar{D}_2\cap\bar{D}_3,
\\
\begin{pmatrix}1 & \overline{h^b(\bar{k})} e^{-2i(kx+2k^2 t)} \\
0 & 1 \end{pmatrix},&k\in\bar{D}_3\cap\bar{D}_4,
\end{cases}
\end{align}
with \(r_1^b=0\) and \(h^b=r^b\), where \(r^b\) is given by \eqref{rb} (see also \cite{tperiodicI}).
This jump matrix is similar to the one given in Theorem \ref{maintheorem}. 
Furthermore, the function \(r_+^b\) satisfies Assumption \ref{assumptions-spectral-functions-branch} as well as Assumption \ref{assumptions-spectral-functions} (c) and (d). 
However,  \(r_+^b(k)\) behaves like \(\frac1k\) as \(k\to\infty\). 
This is expected since the initial and boundary data are not compatible at \(0\).
In order to ensure that the solutions generated by our assumptions have the desired regularity, we thus modify the spectral data \(r^b_+\), \(r_1^b\) and \(h^b\) appearing in \(v^b\) such that \(r^b_+\) has enough decay at infinity (cf. \eqref{seriesexpansionrinfinity}) and such that it vanishes at \(\kappa_+\).
\end{remark}

\begin{remark}[Construction of Admissible Data] 
We will give an example of independent spectral data satisfying Assumptions \ref{assumptions-spectral-functions} and \ref{assumptions-spectral-functions-branch}.
Note that the RH problem with jump matrix \(v^b\) (cf. \eqref{jumpmatrixmodelproblem}) is equivalent (after a simple transformation) to the RH problem with jump matrix \(v^b\) but with  \(r_1^b=0\) and \(h^b=r^b\) replaced by \(r_1^b=\tau r^b\) and \(h^b=(1-\tau) r^b\) for fixed \(\tau\in[0,1]\).
Now let \(r_{\infty}\colon\R\to\C\) be a smooth function such that \(r_{\infty}\) has the same asymptotic behaviour as \(r^b\) at \(\infty\) up to and including order \(4\), \(r_{\infty}(\kappa_+)=r^b(\kappa_+)\), \(r_\infty\) vanishes in a neighbourhood of \([E_1,E_2]\) and such that \(|r-r_\infty|<1\) on \((-\infty,\kappa_+]\setminus[E_1,E_2]\).
Then the data \(h=\tau r^b\) and \(r_1=(1-\tau)r^b-r_\infty\) satisfy Assumptions \ref{assumptions-spectral-functions} and \ref{assumptions-spectral-functions-branch} for every \(\tau\in[0,1]\).
\end{remark}

The proof of Theorem \ref{maintheorem} is divided into two parts. 
In the first part, we show using the steepest descent analysis that the RH problem \eqref{RHm} has a solution and that the limit \eqref{ulim} exists. 
Moreover, we will derive the asymptotics of this solution from which the long time asymptotics of \(u\) and \(u_x\) follow. 
In the second part, we will then show that the function \(u\) defined by \eqref{ulim} is indeed a solution of \eqref{nls} by applying the dressing type arguments \cite{ZakharovShabat1,ZakharovShabat2}.

\section{Transformations of the Riemann--Hilbert problem}\label{sectiontransformationRHP}

Suppose the triple \((\alpha,\omega,c)\) belongs to the family \eqref{parameterfamily} and suppose that \(r_1\), \(h\) and \(r=r_1+h\) are functions satisfying Assumption \ref{assumptions-spectral-functions} and \ref{assumptions-spectral-functions-branch}. Suppose furthermore that \(c_0\in (0,4\beta-2\alpha)\) is given.

In order to show existence as well as determine the long time asymptotics of the solution $m$ of the RH problem \eqref{RHm}, we will transform the RH problem into a small norm RH problem.
Starting with $m$ we will define functions $m^{(j)}(x,t,k)$, $j = 1, 2,3,4$, such that each $m^{(j)}$ is analytic in \(\C\setminus\Gamma^{(j)}\) and satisfies the jump condition
\begin{align}\label{RHmj}
m^{(j)}_+(x,t,k) = m^{(j)}_-(x, t, k) v^{(j)}(x, t, k) \quad \text{for a.e.} \ k \in \Gamma^{(j)},
\end{align}
where the contours $\Gamma^{(j)}$ and jump matrices $v^{(j)}$ are specified below.

\subsection{First transformation}\label{firsttransformation-new}

Our first step is to introduce a \(g\)-function in order to normalize the oscillatory and exponentially large factors in the jump matrix.
Since our goal is to construct solutions approaching the background plane wave solution $u^b(x,t)$ defined in \eqref{backgroundsolution}, we choose (cf. equation (5.8) in \cite{BIK2009})
\begin{align}\label{gdef}
g(k) \equiv g(\zeta, k) = \Omega(k) + \zeta X(k) = (2k-2\beta+\zeta)X(k).
\end{align}
The differential
$$dg(k) = \frac{(4k+\zeta)(\beta+k) - 2\alpha^2}{X(k)}dk$$
can be written as
\begin{align}\label{gfunction-derivative}
dg(k) = 4\frac{(k-k_1)(k-k_0)}{X(k)}dk,
\end{align}
where the zeros $k_i \equiv k_i(\zeta) \in \R$, \(i=0,1\), are located at
$$k_0 = -\frac{4\beta + \zeta}{8} + \sqrt{\frac{\alpha^2}{2} + \Big(\frac{4\beta - \zeta}{8}\Big)^2},\quad k_1 = -\frac{4\beta + \zeta}{8} - \sqrt{\frac{\alpha^2}{2} + \Big(\frac{4\beta - \zeta}{8}\Big)^2}.$$
For $\zeta = 0$, we have
$$E_1 < k_1 < E_2 < \kappa_+ = k_0.$$
As $\zeta$ increases, $k_0$ moves to the left until it hits $E_2$ for $\zeta = 4\beta - 2\alpha$. This determines the right boundary of the plane wave sector. 

%
The signature table of $\im g$ for \(0\leq \zeta<4\beta-2\alpha\) 
 is shown in Figure \ref{signature_table_Img.pdf}. 
 The non-horizontal branch of the level set where \(\im g=0\) asymptotes to the vertical line \(\re k=-\zeta/4\).
%
%
%
\begin{figure}
\begin{center}
 \begin{overpic}[width=.7\textwidth]{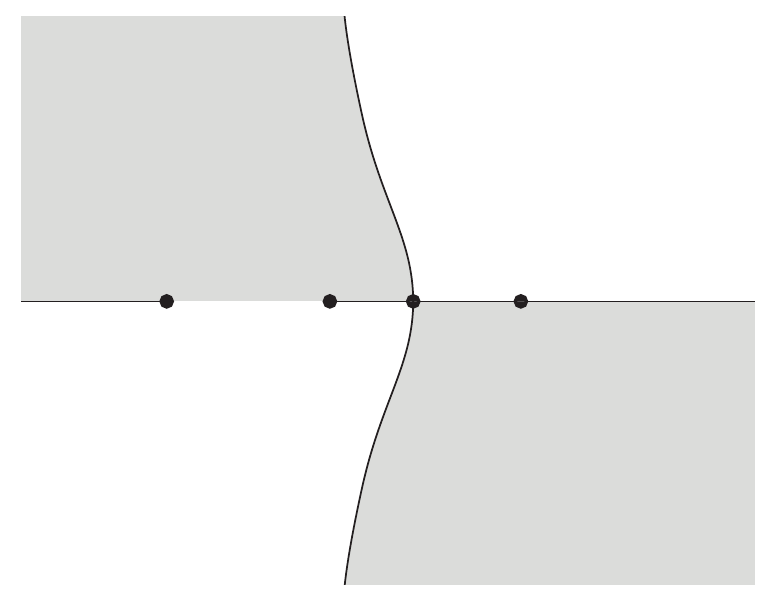}
      \put(20,34){\small $E_1$}
      \put(41,34){\small $E_2$}
      \put(66,34){\small $\kappa_+$}
      \put(54,34){\small $k_0$}
      \put(12,61){\small $\im g<0$}
      \put(73,14){\small $\im g<0$}
      \put(12,14){\small $\im g>0$}
      \put(73,61){\small $\im g>0$}
    \end{overpic}
     \begin{figuretext}\label{signature_table_Img.pdf}
     The signature table of \(\im g\) for \(0\leq \zeta<4\beta-2\alpha\).   The region where \(\im g<0\) is shaded. The solid line represents the level set where \(\im g=0\).
         \end{figuretext}
     \end{center}
\end{figure}

\begin{lemma}\label{glemma}
Let \(0\leq\zeta<4\beta-2\alpha\).
The function $g(k)$ defined in (\ref{gdef}) has the following properties:
\begin{enumerate}[$(a)$]
\item $g(k) - 2k^2- \zeta k$ is an analytic and bounded function of $k \in \C \setminus [E_1,E_2]$ with continuous boundary values on $(E_1, E_2)$.

\item $g(k)$ obeys the symmetry 
\begin{align}\label{gsymm}
  g(k) = \overline{g(\bar{k})}, \qquad k \in \C\setminus [E_1, E_2].
\end{align}

\item $g(k)$ satisfies the jump condition
\begin{align}\label{gjump}
  g_+(k) + g_-(k) = 0, \qquad k \in [E_1, E_2].
\end{align}

\item $g(k)$ satisfies
\begin{align*}
g(k) = 2k^2 + \zeta k + g_\infty(\zeta) + O(k^{-1}), \qquad k \to \infty,
\end{align*}
where
\begin{align}\label{defginfty}
g_\infty(\zeta) = \beta \zeta - \alpha^2 - 2\beta^2.
\end{align}
\end{enumerate}
\end{lemma}

Now define $m^{(1)}(x,t,k)$ by
$$m^{(1)}(x,t,k) = e^{-itg_\infty(\zeta)\sigma_3} m(x,t,k) e^{it( g(k) - 2k^2-\zeta k )\sigma_3}.$$
Then $m^{(1)}$ satisfies the jump condition (\ref{RHmj}) with $j = 1$, where $\Gamma^{(1)} = \Gamma$ and the jump matrix $v^{(1)}$ is given by
$$v^{(1)}(x,t,k) = e^{-it(g_-(k)-2k^2-\zeta k)\sigma_3} v(x,t,k) e^{it(g_+(k)-2k^2-\zeta k)\sigma_3}.$$
Using \eqref{gjump} as well as the assumption that \(|r|=1\) on \([E_1,E_2]\), we find
\begin{align*}
& v_1^{(1)} = \begin{pmatrix}1 - |r_1(k)|^2 & \overline{r_1(k)} e^{-2itg(k) } \\
-r_1(k) e^{2itg(k)} & 1 \end{pmatrix},
\qquad 
v_2^{(1)} = \begin{pmatrix}1 & 0 \\
-h(k) e^{2itg(k)} & 1 \end{pmatrix},
	\\ 
& v_3^{(1)} = \begin{pmatrix}1 - |r(k)|^2 & \overline{r(k)} e^{-2itg(k) } \\
-r(k) e^{2itg(k)} & 1 \end{pmatrix},
\qquad 
v_4^{(1)} = \begin{pmatrix}1 & \overline{h(\bar{k})} e^{-2itg(k)} \\
0 & 1 \end{pmatrix},
	\\
& v_5^{(1)} = \begin{pmatrix}0& \overline{r(k)} \\
-r(k) & e^{-2itg_+(k)} \end{pmatrix},
\end{align*}
where \(v^{(1)}_j\) for \(j\in\{1,2,3,4,5\}\) denotes the restriction of \(v^{(1)}\) to the subcontour of \(\Gamma^{(1)}=\Gamma\) labeled by \(j\) in Figure \ref{Gamma.pdf}.

\subsection{Second transformation}\label{secondtransformation-new}

\begin{figure}
\begin{center}
 \begin{overpic}[width=.7\textwidth]{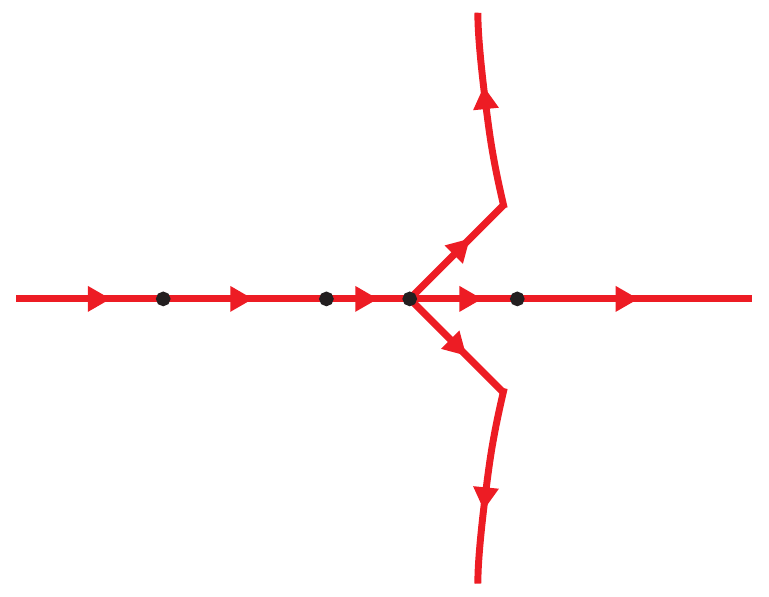}
      \put(101,37.5){\small $\Gamma^{(2)}$}
      \put(19,34){\small $E_1$}
      \put(40,34){\small $E_2$}
      \put(67,35){\small $\kappa_+$}
      \put(52,34){\small $k_0$}
      \put(58,47){\small $2$}
      \put(58,28){\small $4$}
      \put(81,34){\small $1$}
      \put(60,65){\small $2$}
      \put(60,10){\small $4$}
      \put(12,34){\small $3$}
      \put(47,34){\small $3$}
      \put(31,34){\small $5$}
      \put(61,34){\small $1$}
    \end{overpic}
     \begin{figuretext}\label{Gamma2-minimal.pdf}
        The contour  $\Gamma^{(2)}$ in the complex $k$-plane. 
         \end{figuretext}
     \end{center}
\end{figure}

The purpose of the second transformation is to deform the non-horizontal part of the contour \(\Gamma=\Gamma^{(1)}\) so that it passes through the critical point \(k_0\). The new contour \(\Gamma^{(2)}\) is shown in Figure \ref{Gamma2-minimal.pdf}.
Define $m^{(2)}(x,t,k)$ by
\begin{align*}
m^{(2)}=m^{(1)}H_2,
\end{align*}
where
\begin{align*}
H_2(\zeta,k)= \begin{cases}\begin{pmatrix}1&0\\h(k)e^{2itg(k)}&1\end{pmatrix},& k\in V_1,
	\\
	\begin{pmatrix}1&\overline{h(\bar{k})}e^{-2itg(k)}\\0&1\end{pmatrix},& k\in V_2,
	\\
	I,&\textrm{elsewhere},
	\end{cases}
\end{align*}
with the domains \(V_1\) and \(V_2\) given as in Figure \ref{Ujs-hjumpdeformation.pdf}.
\begin{figure}
\begin{center}
 \begin{overpic}[width=.7\textwidth]{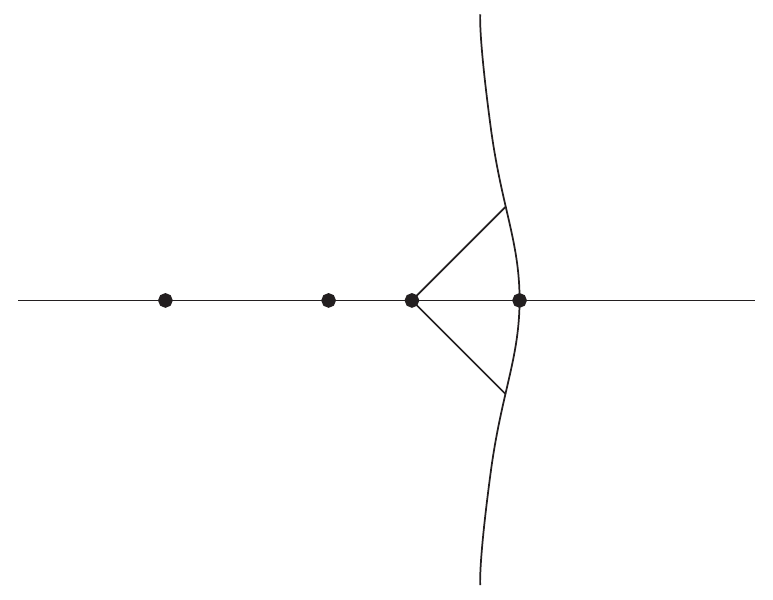}
      \put(102,37.5){\small $\re k$}
      \put(18,33){\small $E_1$}
      \put(40,33){\small $E_2$}
      \put(69,35){\small $\kappa_+$}
      \put(51,34){\small $k_0$}
      \put(61,42){\small $V_1$}
      \put(61,33){\small $V_2$}
    \end{overpic}
     \begin{figuretext}\label{Ujs-hjumpdeformation.pdf}
        The open subsets $V_1$ and $V_2$ of the complex $k$-plane. 
         \end{figuretext}
     \end{center}
\end{figure}
Then $m^{(2)}$ satisfies the jump condition (\ref{RHmj}) with $j = 2$, where the jump $v^{(2)}$ across \(\Gamma^{(2)}\) is given by
\begin{align*}
& v_1^{(2)} = \begin{pmatrix}1 - |r_1(k)|^2 & \overline{r_1(k)} e^{-2itg(k)} \\
-r_1(k) e^{2itg(k)} & 1 \end{pmatrix},
\qquad 
v_2^{(2)} = \begin{pmatrix}1 & 0 \\
-h(k) e^{2itg(k)} & 1 \end{pmatrix},
	\\ 
& v_3^{(2)} = \begin{pmatrix}1 - |r(k)|^2 & \overline{r(k)} e^{-2itg(k)} \\
-r(k) e^{2itg(k)} & 1 \end{pmatrix},
\qquad 
v_4^{(2)} = \begin{pmatrix}1 & \overline{h(\bar{k})} e^{-2itg(k)} \\
0 & 1 \end{pmatrix},
	\\
& v_5^{(2)} = \begin{pmatrix}0& \overline{r(k)} \\
-r(k) & e^{-2itg_+(k)} \end{pmatrix},
\end{align*}
and the subscripts refer to Figure \ref{Gamma2-minimal.pdf}.
To find \(v_{10}^{(2)}\) we have used that \(r=r_1+h\) on \((-\infty,\kappa_+]\).

\subsection{Third transformation}\label{thirdtransformation-new}

The jump matrix \(v_3^{(2)}\) has the wrong factorization. Hence we introduce $m^{(3)}(x,t,k)$ by 
$$m^{(3)} = D(\zeta,\infty)^{\sigma_3} m^{(2)} D(\zeta,k)^{-\sigma_3},$$
where
\begin{align}
D(\zeta,\infty)&= e^{-\frac{1}{2\pi i}\big\{\big(\int_{-\infty}^{E_1} + \int_{E_2}^{k_0}\big) \frac{\ln(1 - |r(s)|^2)}{X(s)} ds
+ \int_{E_1}^{E_2} \frac{i \arg(r(s))}{X_+(s)} ds\big\}},
\label{Dinfinitydefinition}
\\\label{Ddef}
D(\zeta, k) &= e^{\frac{X(k)}{2\pi i}\big\{\big(\int_{-\infty}^{E_1} + \int_{E_2}^{k_0}\big) \frac{\ln(1 - |r(s)|^2)}{X(s)(s - k)} ds
+ \int_{E_1}^{E_2} \frac{i\arg(r(s))}{X_+(s)(s - k)} ds\big\}},
\end{align}
for \(\zeta\in[0,c_0]\) and  \(k\in\C\setminus(-\infty,k_0]\).
Here the function \(\arg\) is chosen as in \eqref{D0infinitydefinition}. Note that this definition is consistent with \eqref{D0infinitydefinition} for \(\zeta=0\).
For \(\epsilon>0\) and \(z\in\C\) let \(D_\epsilon(z)=\{k\in\C\,|\, |z-k|<\epsilon\}\) denote the open disk of radius \(\epsilon\) around \(z\).

\begin{lemma}\label{Dlemma}
For each \(\zeta\in[0,c_0]\), the function $D(\zeta, k)$ defined in (\ref{Ddef}) has the following properties:
\begin{enumerate}[$(a)$]
\item $D(\zeta,k)$ is an analytic function of $k \in \C \setminus (-\infty, k_0]$ with continuous boundary values on $(-\infty, k_0) \setminus \{E_1, E_2\}$.

\item  $D(\zeta,\cdot)-D(\zeta,\infty)$, $D(\zeta,\cdot)^{-1}-D(\zeta,\infty)^{-1} \in \dot{E}^2(\C \setminus (-\infty, k_0])$. 

\item \(D(\zeta,k)^{\pm1}=D(\zeta,\infty)^{\pm1}+O(k^{-1})\) as \(k\to\infty\)
uniformly with respect to \(\zeta\in[0,c_0]\).
Furthermore, \(|D(\zeta,\infty)|=1\).

\item $D(\zeta,k)$ obeys the symmetry  
\begin{align}\label{Dsymm}
  D(\zeta,k)  \overline{D(\zeta,\bar{k})}=1, \qquad k \in \C \setminus (-\infty, k_0].
\end{align}
In particular, \(D_+(\zeta,k)\overline{D_-(,\zeta,k)}=1\) for \(k\in(-\infty,k_0)\setminus\{E_1,E_2\}\).

\item $D(\zeta,k)$ satisfies the jump conditions
\begin{align*}
&  D_+(\zeta,k) = D_-(\zeta,k)(1 - |r(k)|^2), \qquad k \in (-\infty, E_1) \cup (E_2, k_0),
	\\
&  D_+(\zeta,k)D_-(\zeta,k) = r(k), \qquad k \in (E_1, E_2).
\end{align*}

\item As $k$ approaches the branch points $E_1$ and $E_2$, \(D(\zeta,k)\) exhibits the asymptotics
\begin{align*}
\begin{cases}
D(\zeta,k)=(k-E_1)^{\frac14}e^{\ln(2q_{1,0}q_{1,1})/2+i\arg(i q_{1,0})/2}(1+O(|k-E_1|^{1/2})), &k\to E_1,
	\\
D(\zeta,k)= (E_2-k)^{\frac14}e^{\ln(-2q_{2,0}q_{2,1})/2+i\arg(i q_{2,0})/2}(1+O(|k-E_2|^{1/2})), &k\to E_2,
\end{cases}
\end{align*}
 uniformly for \(\im k\geq 0\) and \(\zeta\in[0,c_0]\).
Here \(z\mapsto z^{\frac14}\) denotes the principal branch of the \(4\)th-root, positive on \((0,\infty)\) with a branch cut along \((-\infty,0]\).

Due to \eqref{Dsymm} similar asymptotic formulas hold for \(\im k\leq 0\).

\item Let \(\epsilon>0\) be given such that \(D_\epsilon(k_0)\cap[E_1,E_2]=\emptyset\) for all \(\zeta\in [0,c_0]\).
Put \(\nu\equiv \nu(r(k_0))\!:=-\frac1{2\pi}\ln(1-|r(k_0)|^2)\) and define
\begin{align*}
D_b(\zeta,k)=(k-k_0)^{-i\nu}D(\zeta,k),\qquad k\in D_\epsilon(k_0)\setminus(-\infty,k_0],
\end{align*}
where the logarithm is defined using the principal branch, real on \((0,\infty)\), with a branch cut along \((-\infty,k_0]\).
Then \(D_b\) is a continuous function of \(k\in D_\epsilon(k_0)\setminus(-\infty,k_0]\) with continuous boundary values along \((-\infty,k_0)\), taking a definite limit as \(k\) approaches \(k_0\) nontangentially. 
Furthermore, \(D_b(\zeta,k)\) and \(D_b(\zeta,k)^{-1}\) are uniformly bounded in \(D_\epsilon(k_0)\) with respect to \(\zeta\in[0,c_0]\).
Additionally, the following estimate is valid:
\begin{align*}
&\left|\frac{D_b(\zeta,k)}{D_b(\zeta,k_0)}-1\right|\leq C(\epsilon)|k-k_0|(1+|\ln|k-k_0||)
\end{align*}
for \(k\in D_\epsilon(k_0)\)and  \(\arg(k_0-k)\in \{\pm\frac\pi4,\pm\frac{3\pi}4\}\), uniformly in \(\zeta\in [0,c_0]\).

\item For every \(\epsilon>0\), the functions \(D\) and \(D^{-1}\) are uniformly bounded on \(\C\setminus (D_\epsilon(E_1)\cup D_\epsilon(E_2)\cup D_\epsilon(k_0(\zeta)))\) with respect to \(\zeta\in[0,c_0]\).
\end{enumerate}
\end{lemma}


\begin{proof}
The proof is standard so we will only present a sketch.
The statements  \textit{(a)-(e)} follow from a detailed study of the involved Cauchy-type integrals. 
In particular, \textit{(d)} and \textit{(e)} are a consequence of the Sokhotski--Plemelj theorem.

For \textit{(f)} we will only consider the case of the branch point \(E_2\). The branch point \(E_1\) can be treated similarly. Choose points \(E_l\) and \(E_u\) such  that \(E_1<E_l<E_2<E_u<k_0\).
From the expansion \eqref{seriesexpansionr} together with the series expansions of the logarithm and the root \(X\), it follows that 
\begin{align*}
\frac{1}{2\pi i}\int_{E_2}^{E_u} \frac{\ln(1 - |r(s)|^2)}{X(s)} \frac{ds}{s-k}
	&=\frac{\ln(-2q_{2,0}q_{2,1})}{2\pi i\sqrt{E_2-E_1}}\int_{E_2}^{E_u}\frac1{\sqrt{s-E_2}}\frac{ds}{s-k}\nonumber
	\\
	&\quad+\frac1{4\pi i\sqrt{E_2-E_1}}\int_{E_2}^{E_u}\frac{\ln(s-E_2)}{\sqrt{s-E_2}}\frac{ds}{s-k}\nonumber
	\\
	&\quad+\frac{2q_{2,0}q_{2,2}+q_{2,1}^2}{4\pi iq_{2,0}q_{2,1}\sqrt{E_2-E_1}}\int_{E_2}^{E_u}\frac{ds}{s-k}\nonumber
	\\
	&\quad + \int_{E_2}^{E_u} f(s)\frac{ds}{s-k},
\end{align*}
where \(f\) is a H{\"o}lder continuous  function vanishing at \(E_2\).
A standard study of the integrals appearing on the right hand side above, as presented for example in \cite{Muskhelishvili}, shows that
\begin{align*}
\frac{1}{2\pi i}\int_{E_2}^{E_u} \frac{\ln(1 - |r(s)|^2)}{X(s)} \frac{ds}{s-k}
	&=\frac{\ln(-2q_{2,0}q_{2,1})}{2\sqrt{E_2-E_1}}\frac1{\sqrt{k-E_2}}\nonumber
	\\
	&\quad+\frac1{4\sqrt{E_2-E_1}}\frac{\ln(k-E_2)-i\pi}{\sqrt{k-E_2}}\nonumber
	\\
	&\quad-\frac{2q_{2,0}q_{2,2}+q_{2,1}^2}{4\pi iq_{2,0}q_{2,1}\sqrt{E_2-E_1}}\ln|k-E_2|\nonumber
	\\
	&\quad + R_1(k)
\end{align*}
in a neighbourhood of \(E_2\), where the function \(R_1(k)\) is continuous in \(\overline{\C_+}\)
(see Chapter 1, section 8.6  in \cite{Gakhov} for the power-logarithmic type integral, Chapter 1, \textsection 16, in \cite{Muskhelishvili} for the remainder term and Chapter 4, \textsection 29, in \cite{Muskhelishvili} for the remaining terms).

A similar calculation shows that
\begin{align*}
\frac1{2\pi i}\int_{E_l}^{E_2} \frac{\ln(r(s))}{X_+(s)} \frac{ds}{s-k}
&=\frac{\ln(iq_{2,0})}{2 \sqrt{E_2-E_1}}\frac{1}{\sqrt{k-E_2}}
\\&\quad+\frac{q_{2,1}}{2\pi i q_{2,0}\sqrt{E_2-E_1}}\ln|k-E_2|+R_2(k)
\end{align*}
for \(\im k\geq 0\), where the function \(R_2(k)\) is continuous for \(\im k\geq 0\).

Next we note that near \(E_2\) we have
\begin{align*}
X(k)=\sqrt{k-E_2}(\sqrt{E_2-E_1}+\phi(k)),\quad \im k\geq 0,
\end{align*}
where we take the principal branch of the root with a branch cut along \((-\infty,E_2]\) and \(\phi\) is an analytic function in a neighbourhood of \(E_2\) and satisfies \(|\phi(k)|=O(|k-E_2|)\).
Combining the above calculations, it follows that
\begin{align*}
&\frac{X(k)}{2\pi i}\int_{E_l}^{E_2}\frac{\ln(r(s))}{X_+(s)}\frac{ds}{s-k}+\frac{X(k)}{2\pi i}\int_{E_2}^{E_u}\frac{\ln(1-|r(s)|^2)}{X(s)}\frac{ds}{s-k}
\\
&\quad=\frac{\ln(-2q_{2,0}q_{2,1})}2\frac{\sqrt{k-E_2}}{\sqrt{k-E_2}}+\frac14(\ln(k-E_2)-i\pi)\frac{\sqrt{k-E_2}}{\sqrt{k-E_2}}+
\frac{\ln(iq_{2,0})}{2}\frac{\sqrt{k-E_2}}{\sqrt{E_2-k}}
\\&\qquad+\left(\frac{q_{2,1}}{q_{2,0}}-\frac{2q_{2,0}q_{2,2}+q_{2,1}^2}{2q_{2,0}q_{2,1}}\right)\frac{X(k)\ln|k-E_2|}{2\pi i\sqrt{E_2-E_1}}+R(k)
\end{align*}
for \(\im k\geq 0\) in a neighbourhood of \(E_2\),
where \(R(k)\) is continuous in \(\overline{\C_+}\) and satisfies \(|R(k)|=O(|k-E_2|^{1/2})\).
Assumption \eqref{cancellation_assumption} implies that the mixed logarithmic term vanishes. 
Since the remaining integrals in the exponent of \(D\) define analytic functions near \(E_2\), part \textit{(f)} follows. 

Part \textit{(g)} and \textit{(d)} can be shown using similar arguments.
\end{proof}

It follows that $m^{(3)}(x,t,k)$ satisfies the jump condition
 (\ref{RHmj}) with $j = 3$, where $\Gamma^{(3)} = \Gamma^{(2)}$. 
Let \(f^*(k)\!:=\overline{f(\bar k)}\) denote the Schwartz conjugate of a function \(f(k)\). 
Using the jump relations satisfied by \(D\) from Lemma \ref{Dlemma} we find that \(v^{(3)}=D_-^{\sigma_3}v^{(2)}D_+^{-\sigma_3}\) is given by
\begin{align*}
& v_1^{(3)} = \begin{pmatrix}1 - |r_1|^2 & D^2\overline{r_1} e^{-2itg} \\
-D^{-2}r_1 e^{2itg} & 1 \end{pmatrix},
\qquad 
v_2^{(3)}=  \begin{pmatrix}1 & 0 \\
-D^{-2}h e^{2itg} & 1 \end{pmatrix},
	\\
	& v_3^{(3)} = \begin{pmatrix}1 & D_+D_-\overline{r} e^{-2itg} \\
-D_+^{-1}D_-^{-1}r e^{2itg} & 1 - |r|^2 \end{pmatrix}=\begin{pmatrix}1 & D_+^2\frac{\overline{r}}{1-|r|^2} e^{-2itg} \\
-D_-^{-2}\frac{r}{1-|r|^2} e^{2itg} & 1 - |r|^2 \end{pmatrix},
	\\ 
	& 
v_4^{(3)}= \begin{pmatrix}1 & D^2h^* e^{-2itg} \\
0 & 1 \end{pmatrix},
\qquad v_5^{(3)} = \begin{pmatrix}0& 1 \\
-1 & D_+D_-^{-1}e^{-2itg_+} \end{pmatrix},
\end{align*}
with the subscripts refering to Figure \ref{Gamma2-minimal.pdf}.
Define
\begin{align*}
 r_2(k)=
	\begin{cases}
		\frac{\overline{r(k)}}{1-|r(k)|^2},&k\in(-\infty,\kappa_+]\setminus[E_1,E_2],
			\\
		-\frac{r(k)}{1+r(k)^2},&k\in(E_1,E_2).
	\end{cases}
\end{align*}
Then we can write the jumps across the real axis excluding the branch cut \([E_1,E_2]\) as
\begin{align*}
&v_1^{(3)}= \begin{pmatrix}1 & D^2\overline{r_1} e^{-2itg} \\
0 & 1 \end{pmatrix}\begin{pmatrix}1 & 0\\
-D^{-2}r_1 e^{2itg} & 1 \end{pmatrix},
	\\
& v_3^{(3)} =\begin{pmatrix}1 & 0 \\
-D_-^{-2}\overline{ r_2} e^{2itg} & 1  \end{pmatrix}\begin{pmatrix}1 & D_+^2 r_2 e^{-2itg} \\
0 & 1 \end{pmatrix}.
\end{align*}

\subsection{Fourth transformation}\label{fourthtransformation-new}

The goal of the fourth transformation is to deform the contour in such a way that the jump matrix contains the exponential factors  \(e^{2itg}\) and \(e^{-2itg}\) on the parts of the contour where \(\im g\) is positive and negative, respectively.
We first need to introduce analytic approximations of the functions \(r_1\), \(r_2\) and \(h\).  Following \cite{DeiftZhou}, we split each function into an analytic part and a small remainder. The remainder will remain on the original contour while the analytic part will be deformed.

Let \(U_i\), \(i=1,\dots,6\), be the domains displayed in Figure \ref{Ujs.pdf-new}.

\begin{figure}
\begin{center}
 \begin{overpic}[width=.7\textwidth]{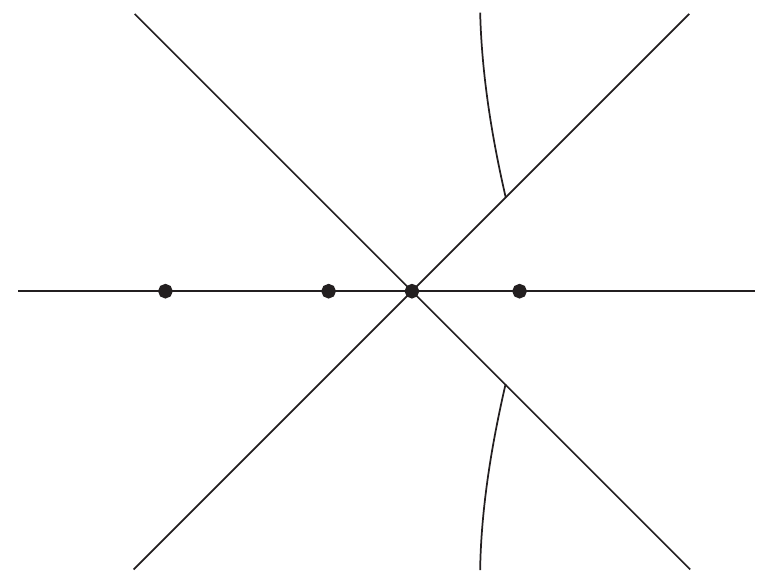}
      \put(102,36.5){\small $\re k$}
      \put(19.5,33){\small $E_1$}
      \put(40.5,33){\small $E_2$}
      \put(66,34){\small $\kappa_+$}
      \put(52,33){\small $k_0$}
      \put(82,51){\small $U_1$}
      \put(68,64){\small $U_6$}
      \put(18,51){\small $U_2$}
      \put(18,21){\small $U_3$}
      \put(68,9){\small $U_5$}
      \put(82,21){\small $U_4$}
    \end{overpic}
     \begin{figuretext}\label{Ujs.pdf-new}
        The open subsets $\{U_j\}$ of the complex $k$-plane. 
         \end{figuretext}
     \end{center}
\end{figure}

\begin{lemma}[Analytic approximation of \( r_2\)]\label{analyticapproximation r2}
There exists a decomposition
\begin{align*}
 r_2(k)=r_{2,a}(x,t,k)+r_{2,r}(x,t,k),\quad k<k_0,
\end{align*}
where the functions \(r_{2,a}\) and \(r_{2,r}\) have the following properties:
\begin{enumerate}[(a)]
\item For each \(\zeta\in [0,c_0]\) and each \(t>0\), the function \(r_{2,a}(x,t,k)\) is defined and continuous for \(k\in \bar U_2\setminus\{E_1,E_2\}\) and analytic for \(k\in U_2\).

\item For each \(\epsilon>0\), there exists a constant \(C(\epsilon)\) such that the function \(r_{2,a}\) satisfies
\begin{align}\label{ r2aestimatefirst}
&|r_{2,a}(x,t,k)- r_2(k_0)|\leq C(\epsilon)|k-k_0|e^{\frac{t}4|\im g(\zeta,k)|}, \qquad k\in D_\epsilon(k_0), 
\end{align}
as well as
\begin{align}
&|r_{2,a}(x,t,k)|\leq \frac{C(\epsilon)}{1+|k|}e^{\frac{t}4|\im g(\zeta,k)|},
\quad k\in\bar U_2\setminus (D_{\epsilon}(E_1)\cup D_{\epsilon}(E_2)),\label{ r2aestimatesecond}
\end{align}
uniformly in \(\zeta\in[0,c_0]\) and \(t>0\).

\item The \(L^1\), \(L^2\), and \(L^\infty\)-norms on \((-\infty,k_0)\setminus[E_1,E_2]\) of \(k\mapsto(1+|k|^2)r_{2,r}(x,t,k)\) are \(O(t^{-3/2})\) as \(t\to\infty\) uniformly with respect to \(\zeta\in[0,c_0]\) and the \(L^1\), \(L^2\), and \(L^\infty\)-norms on \((E_1,E_2)\) of \(|r_{2,r}(x,t,\cdot)|e^{-2t|\im g_+(\zeta,\cdot)|}\) are \(O(t^{-3/2})\) as \(t\to\infty\) uniformly with respect to \(\zeta\in[0,c_0]\).
\end{enumerate}
\end{lemma}

\begin{proof}

We first show that there exists a function \(f_0(\zeta,k)\) which is analytic in \(U_2\), continuous on \(\bar{U}_2\setminus\{E_1,E_2\}\), and satisfies

\begin{enumerate}[i)]
\item
$\begin{aligned}[t]\label{asymptoticsf}
\frac{\partial^n}{\partial k^n}( r_2-f_0)=
	\begin{cases}
		O((k-k_0)^{6-n}),&k\to k_0,k\leq k_0,
			\\
		O(k^{-5-n}),&k\to-\infty,k\in\R,
			\\
		O((k-E_i)^{3-n}),&k\to E_i,k\in\R,i=1,2,
	\end{cases}
	\quad n=0,1,2,
\end{aligned}$
\\
with uniform error estimates with respect to \(\zeta\in [0,c_0]\);
\item
$\begin{aligned}[t]\label{estimatef0}
|f_0(x,t,k)|\leq \frac{C(\epsilon)}{1+|k|},
\quad k\in\bar U_2\setminus (D_{\epsilon}(E_1)\cup D_{\epsilon}(E_1)) ,\zeta\in[0,c_0],t>0.
\end{aligned}$
\end{enumerate}

Note that since \(r\in C^6((E_2,\kappa_+])\), it admits a Taylor series expansion at \(k_0\) of the form
\begin{align}\label{seriesexpansionr3}
r(k)=
	\sum_{j=0}^5w_j(\zeta)(k-k_0)^j+O((k-k_0)^6),\qquad k\to k_0,
\end{align}
where the coefficients \(w_j(\zeta)\!:={r^{(j)}(k_0)}/{j!}\in \C\) are uniformly bounded with respect to \(\zeta\in[0,c_0]\) and the error term is uniform with respect to \(\zeta\in[0,c_0]\).
The expansion \eqref{seriesexpansionr3} together with assumptions \eqref{seriesexpansionr} and \eqref{seriesexpansionrinfinity} show - after a long but straightforward calculation - that \( r_2\) admits the series expansions 
\begin{align*}
 r_2(k)=
	\begin{cases}
	\sum_{l=-1}^5{Q_{2,l}}(k-E_2)^{l/2}+O((k-E_2)^{3}),&k\downarrow E_2,
		\\ 
	\sum_{l=-1}^5i^l{Q_{2,l}}(E_2-k)^{l/2}+O((E_2-k)^{3}),&k\uparrow E_2,
		\\
	\sum_{l=-1}^5i^l{Q_{1,l}}(k-E_1)^{l/2}+O((k-E_1)^{3}),&k\downarrow E_1,
		\\ 
	\sum_{l=-1}^5(-1)^l{Q_{1,l}}(E_1-k)^{l/2}+O((E_1-k)^{3}),&k\uparrow E_1,
		\\
	O(k^{-5}),&k\to-\infty,
		\\ 
	\sum_{l=0}^5{W_l(\zeta)}(k-k_0)^l+O((k-k_0)^6),&k\to k_0,
	\end{cases}
\end{align*}
and these expansions can be differentiated termwise two times.
The coefficients \(Q_{i,l}\) are rational functions in \(q_{i,l}\) with a denominator of the form \(q_{i,0}^{l+1}q_{i,1}^{l+2}\). Since \(q_{i,0},q_{i,1}\neq 0\) by assumption, these coefficients are well defined and finite. 
Similarly, the coefficients \(W_l(\zeta)\) are rational functions in \(w_j\) with a denominator of the form \(1-|w_0(\zeta)|^2=1-|r(k_0)|^2\). Again, the assumption \(|r|<1\) on \((E_2,\kappa_+)\) implies that the \(W_l(\zeta)\) are well defined and finite.
Moreover, the coefficients \(Q_{i,l}\) and  \(W_l(\zeta)\) are uniformly bounded  with respect to \(\zeta\in[0,c_0]\) and the error terms in the above series expansions are uniform with respect to \(\zeta\in[0,c_0]\).
To find \(f_0\) we write
\begin{align*}
f_0(\zeta,k)=\frac{p_1(\zeta,k)}{X(k)}+p_2(\zeta,k)
\end{align*}
with
$$p_1(\zeta,k)=\sum_{j=4}^{17}\frac{a_j(\zeta)}{(k+i)^j}\quad\text{and}\quad p_2(\zeta,k)=\sum_{j=5}^{16}\frac{b_j(\zeta)}{(k+i)^j},$$ where the coefficients \(a_j(\zeta)\) and \(b_j(\zeta)\) are chosen such that 
\begin{align*}
\frac{p_1(\zeta,k)}{X_+(k)}=
	\begin{cases}
	\sum_{l=-1,\text{$l$ odd}}^5 Q_{2,l}(k-E_2)^{l/2}+O((k-E_2)^{3}),&k\downarrow E_2,
		\\ 
	\sum_{l=-1,\text{$l$ odd}}^5i^l Q_{2,l}(E_2-k)^{l/2}+O((E_2-k)^{3}),&k\uparrow E_2,
		\\
	\sum_{l=-1,\text{$l$ odd}}^5i^l{Q_{1,l}}(k-E_1)^{l/2}+O((k-E_1)^{3}),&k\downarrow E_1,
		\\ 
	\sum_{l=-1,\text{$l$ odd}}^5(-1)^l{Q_{1,l}}(E_1-k)^{l/2}+O((E_1-k)^{3}),&k\uparrow E_1,
		\\
	O(k^{-5}),&k\to-\infty,
		\\ 
	\frac12\sum_{l=0}^5W_l(\zeta)(k-k_0)^l+O((k-k_0)^6),&k\to k_0,
	\end{cases}
\end{align*}
and
\begin{align*}
p_2(\zeta,k)=
	\begin{cases}
	\sum_{l=-1,\text{$l$ even}}^5 Q_{2,l}(k-E_2)^{l/2}+O((k-E_2)^{3}),&k\downarrow E_2,
		\\ 
	\sum_{l=-1,\text{$l$ even}}^5i^l Q_{2,l}(E_2-k)^{l/2}+O((E_2-k)^{3}),&k\uparrow E_2,
		\\
	\sum_{l=-1,\text{$l$ even}}^5i^l{Q_{1,l}}(k-E_1)^{l/2}+O((k-E_1)^{3}),&k\downarrow E_1,
		\\ 
	\sum_{l=-1,\text{$l$ even}}^5(-1)^l{Q_{1,l}}(E_1-k)^{l/2}+O((E_1-k)^{3}),&k\uparrow E_1,
		\\
	O(k^{-5}),&k\to-\infty,
		\\ 
	\frac12\sum_{l=0}^5W_l(\zeta)(k-k_0)^l+O((k-k_0)^6),&k\to k_0,
	\end{cases}
\end{align*}
with uniform error terms with respect to \(\zeta\in [0,c_0]\) and such that we may differentiate the above expansions twice (see the proof of Lemma 4.5 in \cite{NonLinSteepDescent-Lenells} for details).
The coefficients \(a_j\) and \(b_j\) are polynomial in the coefficients of the above series expansions.
Thus \(a_j\) and \(b_j\) are  uniformly bounded with respect to \(\zeta\in [0,c_0]\).
This proves \eqref{asymptoticsf} and \eqref{estimatef0}.

\begin{figure}
\begin{center}
 \begin{overpic}[width=.7\textwidth]{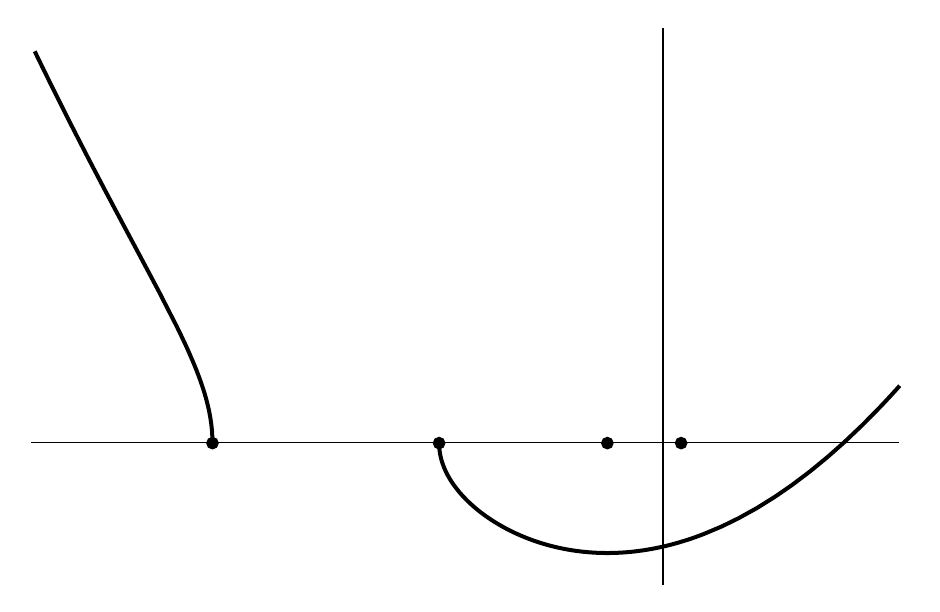}
      \put(20,13.5){\small $E_1$}
      \put(41.5,13.5){\small $E_2$}
      \put(63.5,13.5){\small $k_0$}
       \put(71.5,13.5){\small $\kappa_+$}
       \put(97,16.5){\small $k$}
       \put(72,60){\small $g(\zeta,k)$}
    \end{overpic}
     \begin{figuretext}\label{plot_g_along_R.pdf}
        The graph of \(g\) along \((-\infty,E_1]\) and \([E_1,\infty)\) for a particular choice of \(\alpha\), \(\beta\) and \(\zeta\).
         \end{figuretext}
     \end{center}
\end{figure}

Now define \(f:=r_2-f_0\) on \((-\infty,k_0]\).
For each \(\zeta\in[0,c_0]\) the map \(k\mapsto \phi=g(\zeta,k)\) is a decreasing bijection from \((-\infty,E_1)\to (0,\infty)\) and from \((E_2,k_0)\to (g(\zeta,k_0),0)\), see Figure \ref{plot_g_along_R.pdf}.
Hence we may define a function \(F(\zeta,\phi)\) by
\begin{align*}
F(\zeta,\phi)=
	\begin{cases}
		\frac{(k+i)^6\sqrt{k+i}}{(k-E_1)(k-E_2)(k-k_0)}f(\zeta,k),&\phi\in  (g(\zeta,k_0),0)\cup (0,\infty),
			\\
		0,&\phi\in(-\infty,g(\zeta,k_0))\cup\{0\},
	\end{cases}
\end{align*}
where \(\sqrt{k+i}\) denotes the principal branch of the root with a branch cut along the negative imaginary axis. 
By the chain rule we have
\begin{align}\label{derivativesF}
\frac{\partial^nF}{\partial\phi^n}(\zeta,\phi(\zeta,k))=
		\left(\frac{1}{\partial\phi/\partial k}\frac{\partial}{\partial k}\right)^n\left(\frac{(k+i)^6\sqrt{k+i}}{(k-E_1)(k-E_2)(k-k_0)}f(\zeta,k)\right)
\end{align}
for \(\phi(\zeta,k)\in (g(\zeta,k_0),0)\cup (0,\infty)\).
In view of \eqref{gfunction-derivative} it follows that
\begin{align*}
\frac{\partial^n}{\partial k^n}\frac1{{\partial\phi}/{\partial k}}(\zeta,k)=
	\begin{cases}
		O((k-k_0)^{-1-n}),&k\to k_0, k<k_0,
			\\		
		O(k^{-1-n}),&k\to-\infty,k\in\R,
			\\
		O((k-E_i)^{1/2-n}),&k\to E_i,k\in\R,i=1,2,
	\end{cases}
	\quad n=0,1,
\end{align*}
where the error terms are uniform with respect to  \(\zeta\in[0,c_0]\).
%
%
Together with the asymptotics of \(f\) given by \eqref{asymptoticsf}, formula \eqref{derivativesF} yields
\begin{align*}
\frac{\partial^nF}{\partial\phi^n}(\zeta,\phi(\zeta,k))=
	\begin{cases}
		O((k-k_0)^{5-2n}),&k\to k_0,k\leq k_0,
			\\
		O(k^{-3/2-2n}),&k\to-\infty,k\in\R,
			\\
		O((k-E_i)^{2-n/2}),&k\to E_i,k\in\R,i=1,2,
	\end{cases}
	\quad n=0,1,2.
\end{align*}
This implies \(F(\zeta,\cdot)\in C^2(\R)\).
From the definition \eqref{gdef}  of \(g\) we see that
there exists a constant \(C>0\) independent of \(\zeta\in [0,c_0]\) such that
\begin{align*}
|\phi|^{1/2}\leq C|k|,\quad\text{i.e.}\quad |k|^{-1}\leq C|\phi|^{-1/2}
\end{align*}
for large \(k<0\). Thus we find
\begin{align*}
\frac{\partial^nF}{\partial^n\phi}(\zeta,\phi)=
	\begin{cases}
		O(\phi^{-3/4-n}),& \phi\to+\infty,\phi\in\R,
			\\
		0,&\phi<g(\zeta,k_0),
	\end{cases}
	\quad n=0,1,2,
\end{align*}
with uniform error bounds with respect to \(\zeta\in[0,c_0]\). Since  \(F(\zeta,\cdot)\in C^2(\R)\) it follows that \(F(\zeta,\cdot)\in H^2(\R)\) with
\begin{align*}
\quad \sup_{\zeta\in [0,c_0]}\big\|\frac{\partial^nF}{\partial\phi^n}(\zeta,\cdot)\big\|_{L^2(\R)}<\infty,\quad n=0,1,2.
\end{align*}
Let
\begin{align*}
\hat F(\zeta, s)=\frac{1}{2\pi}\int_\R F(\zeta,\phi)e^{-i\phi s}d\phi,\quad s\in\R,
\end{align*}
denote the Fourier transform of \(F\) (as a function in \(L^2(\R)\)). 
By Plancherel  we have
\begin{align*}
\sup_{\zeta\in[0,c_0]}\|s^n\hat F(\zeta,s)\|_{L^2(ds)}=\frac1{2\pi}\sup_{\zeta\in[0,c_0]}\big\|\frac{\partial^n}{\partial\phi^n}F(\zeta,\phi)\big\|_{L^2(d\phi)}<\infty,\quad n=0,1,2.
\end{align*}
Together with H\"{o}lder's inequality it follows in particular that the \(L^1\)-norm of \(\hat F(\zeta,\cdot)\) is uniformly bounded with respect to \(\zeta\in[0,c_0]\).
The Fourier inversion theorem together with the fact that \(\hat F(\zeta,\cdot)\in L^1(\R)\) yields
\begin{align*}
F(\zeta,\phi)=\int_\R\hat F(\zeta,s)e^{i\phi s}ds,\quad \phi\in\R.
\end{align*}
By the definition of \(F\) we find
\begin{align}\label{representationformulaf}
\int_\R \hat F(\zeta,s)e^{isg(\zeta,k)}ds&=
	\begin{cases}
		\frac{(k+i)^6\sqrt{k+i}}{(k-E_1)(k-E_2)(k-k_0)}f(\zeta,k),&k\in(-\infty,E_1)\cup (E_2,k_0),
			\\
		0,&k\in\{E_1,E_2,k_0\}.
	\end{cases}
\end{align}
Define
\begin{align*}
f_a(\zeta,k)&=\frac{(k-E_1)(k-E_2)(k-k_0)}{(k+i)^6\sqrt{k+i}}\int_{-\infty}^{\frac t4} \hat F(\zeta,s)e^{isg(\zeta,k)}ds,\quad k\in  U_2,
	\\
f_r(\zeta,k)&=f(\zeta,k)-(f_a)_+(\zeta,k),\quad k\in(-\infty,k_0],
\end{align*}
where \((f_a)_+\) denotes the boundary values of \(f_a\) from the upper half-plane.
Since \(|e^{isg(\zeta,k)}|=e^{-s\im g(\zeta,k)}\) and \(\im g(\zeta,k)<0\) for \(k\in U_2\), the integral exists and \(f_a(\zeta,\cdot)\) is analytic in \(U_2\). Furthermore, since \(\im g_+(\zeta,k)\leq 0\) for \(k\in[E_1,E_2]\), the function \(f_a\) can be continuously extended to \(\bar{U}_2\). By \eqref{representationformulaf} we have
\begin{align*}
f_r(\zeta,k)
	&=f(\zeta,k)-f_a(\zeta,k)
	=\frac{(k-E_1)(k-E_2)(k-k_0)}{(k+i)^6\sqrt{k+i}}\int_{\frac t4}^{\infty} \hat F(\zeta,s)e^{isg(\zeta,k)}ds
\end{align*}
for \(k\in(-\infty,E_1)\cup (E_2,k_0)\)
and
\begin{align}\label{frterms}
f_r(\zeta,k)=f(\zeta,k)-\frac{(k-E_1)(k-E_2)(k-k_0)}{(k+i)^6\sqrt{k+i}}\int_{-\infty}^{\frac t4} \hat F(\zeta,s)e^{isg_+(\zeta,k)}ds
\end{align}
for \(k\in[E_1,E_2]\).
We estimate
\begin{align}
|f_a(\zeta,k)|&\leq \frac{|k-E_1||k-E_2||k-k_0|}{(1+|k|^2)^{13/4}}\|\hat F(\zeta,\cdot)\|_{L^1(\R)}\sup_{s\leq t/4}|e^{isg(\zeta,k)}|\nonumber
		\\
		&\leq C\frac{|k-k_0|}{1+|k|^{9/2}}e^{\frac{t}{4}|\im g(\zeta,k)|},\quad k\in \bar{U}_2,\label{faestimate}
		\\
|f_r(\zeta,k)|
		&\leq \frac{|k-E_1||k-E_2||k-k_0|}{(1+|k|^2)^{13/4}}\|s^2\hat F(\zeta,s)\|_{L^2(ds)}\sqrt{\int_{\frac t4}^\infty s^{-4}ds}\nonumber
		\\
		&\leq C\frac{t^{-3/2}}{1+|k|^{7/2}},\quad k\in(-\infty,E_1)\cup(E_2,k_0).\label{frestimateoffcut}
\end{align}
In order to estimate \(f_r=f-(f_a)_+\) on \((E_1,E_2)\) we consider the two terms on the right hand side of \eqref{frterms} separately. For the first term,  the asymptotics of \(f(\zeta,\cdot)\) near \(E_i\) (in combination with the uniform boundedness of \(f\) on a closed subinterval of \((E_1,E_2)\)) yields
\begin{align*}
|f(\zeta,k)|\leq C |k-E_1|^3|k-E_2|^3, \quad k\in(E_1,E_2),\, i=1,2.
\end{align*}
The second term in \eqref{frterms} can be estimated similarly as in \eqref{faestimate}, so that
\begin{align*}
|(f_a)_+(\zeta,k)| 
	&\leq C|k-E_1||k-E_2|e^{\frac t4|\im g_+(\zeta,k)|}, \quad  k\in (E_1,E_2).
\end{align*}
Now define
\begin{align*}
r_{2,a}(x,t,k)&=f_0(\zeta,k)+f_a(\zeta,k),\quad \zeta\in[0,c_0], k\in\bar U_2,
	\\
r_{2,r}(x,t,k)&=f_r(\zeta,k),\quad k\in(-\infty,k_0).
\end{align*}
From estimate \eqref{frestimateoffcut} it follows that  the \(L^\infty\), \(L^1\) and \(L^2\)-norms of \((1+|k|^2)r_{2,r}\) on \((-\infty,k_0)\setminus[E_1,E_2]\) are \(O(t^{-3/2})\).
For the estimates on \([E_1,E_2]\) we note that by the above computations we have
\begin{align*}
|r_{2,r}(x,t,k)|e^{-2t|\im g_+(\zeta,k)|}&\leq C |k-E_1|^3|k-E_2|^3e^{-2t|\im g_+(\zeta,k)|}\\&\quad+C|k-E_1||k-E_2|e^{-\frac{7t}4|\im g_+(\zeta,k)|}
\end{align*}
for \(k\in(E_1,E_2)\).
Since \(|\im g_+(\zeta,k)|\sim|k-E_i|^{1/2}\) as \(k\to E_i\), \(i=1,2\), 
 easy estimates show that the \(L^\infty\), \(L^1\) and \(L^2\)-norms of  \(|r_{2,r}(x,t,\cdot)|e^{-2t|\im g_+(\zeta,\cdot)|}\)  on \((E_1,E_2)\) are \(O(t^{-3/2})\).
To show  \eqref{ r2aestimatefirst} we note that
\begin{align*}
|r_{2,a}(x,t,k)- r_2(k_0)|
	&\leq |f_a(\zeta,k)|+|f_0(\zeta,k)-f_0(\zeta,k_0)|.
\end{align*}
The first term can be estimated using \eqref{faestimate}. In view of the series expansion of \(f_0\) near \(k_0\) it follows that the second term can be estimated by a uniform constant times \(|k-k_0|\).
Estimate \eqref{ r2aestimatesecond} follows from \eqref{faestimate} together with the estimate \eqref{estimatef0} for \(f_0\).
\end{proof}

The proofs of the following two lemmas are similar to (but easier than) the proof of Lemma \ref{analyticapproximation r2} and will be omitted.

\begin{lemma}[Analytic approximation of h]\label{analyticapproximationh}
There exists a decomposition
$$h(k)=h_a(t,k)+h_r(t,k),\quad t>0, k\in\bar{D}_1\cap \bar{D}_2,$$
where the functions \(h_a\) and \(h_r\) have the following properties:
\begin{enumerate}[(a)]
\item For each \(t>0\), the function \(h_{a}(t,k)\) is defined and continuous for \(k\in \bar{D}_1\) and analytic for \(k\in {D}_1\).

\item The function \(h_a\) satisfies
\begin{align}\label{haestimates}
\begin{cases}
|h_a(t,k)-h(\kappa_+)|\leq C|k-\kappa_+|e^{\frac{t}4|\im g(\zeta,k)|},
\\
|h_{a}(t,k)|\leq \frac{C}{1+|k|}e^{\frac{t}4|\im g(\zeta,k)|},
\end{cases}
 k\in \bar{D}_1,\zeta\in[0,c_0],t>0,
\end{align}

\item The \(L^1\), \(L^2\), and \(L^\infty\)-norms of \(k\mapsto (1+|k|^2)h_{r}(t,k)\) on \(\bar{D}_1\cap \bar{D}_2\)  are \(O(t^{-3/2})\) as \(t\to\infty\).
\end{enumerate}
\end{lemma}

\begin{lemma}[Analytic approximation of \(r_1\)]\label{analyticapproximationr1}
There exists a decomposition
\begin{align*}
r_1(k)=r_{1,a}(x,t,k)+r_{1,r}(x,t,k),\quad k\geq k_0,
\end{align*}
where the functions \(r_{a}\) and \(r_{r}\) have the following properties:
\begin{enumerate}[(a)]
\item For each \(\zeta\in [0,c_0]\) and each \(t>0\), the function \(r_{1,a}(x,t,k)\) is defined and continuous for \(k\in \bar U_1\) and analytic for \(k\in U_1\).

\item The function \(r_{1,a}\) satisfies
\begin{align}\label{r1aestimates}
\begin{cases}
|r_{1,a}(x,t,k)-r_1(k_0)|\leq C|k-k_0|e^{\frac{t}4|\im g(\zeta,k)|},
	\\
|r_{1,a}(x,t,k)|\leq \frac{C}{1+|k|}e^{\frac{t}4|\im g(\zeta,k)|},
\end{cases}
 k\in\bar U_1,\zeta\in[0,c_0],t>0,
\end{align}

\item The \(L^1\), \(L^2\), and \(L^\infty\)-norms on \((k_0,\infty)\) of \(k\mapsto (1+|k|^2)r_{1,r}(x,t,k)\) are \(O(t^{-3/2})\) as \(t\to\infty\) uniformly with respect to \(\zeta\in[0,c_0]\).
\end{enumerate}
\end{lemma}


Now let \(\Gamma^{(4)}\) be the contour shown in Figure \ref{Gamma2.pdf-new}.
\begin{figure}
\begin{center}
 \begin{overpic}[width=.7\textwidth]{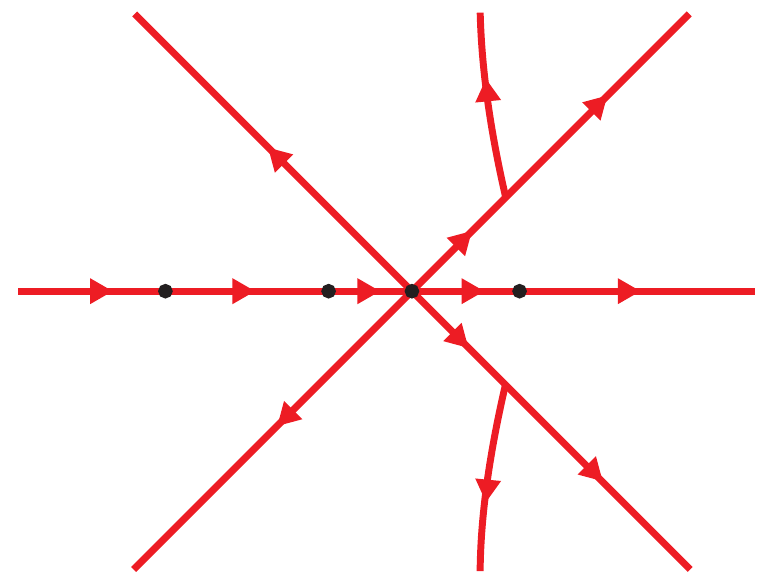}
      \put(102,36.5){\small $\Gamma^{(4)}$}
      \put(19.5,33){\small $E_1$}
      \put(40.5,33){\small $E_2$}
      \put(67,34){\small $\kappa_+$}
      \put(52,33){\small $k_0$}
      \put(57,46){\small $1$}
      \put(32,53){\small $2$}
      \put(33,22){\small $3$}
      \put(57,28){\small $4$}
      \put(79,14){\small $5$}
      \put(79,60){\small $6$}
      \put(81,33.5){\small $7$}
      \put(59,63){\small $8$}
      \put(59,12){\small $9$}
      \put(12,33.5){\small $10$}
      \put(45,33.5){\small $10$}
      \put(30,33.5){\small $11$}
      \put(60.5,33.5){\small $7$}
    \end{overpic}
     \begin{figuretext}\label{Gamma2.pdf-new}
        The contour  $\Gamma^{(4)}$ in the complex $k$-plane. 
         \end{figuretext}
     \end{center}
\end{figure}
Define $m^{(4)}(x,t,k)$ by
\begin{align*}
m^{(4)} &= m^{(3)}D(\zeta,k)^{\sigma_3}H_4D(\zeta,k)^{-\sigma_3}
\end{align*}
where
\begin{align*}
H_4(\zeta,k)=
\begin{cases}
\begin{pmatrix} 1 & 0 \\ r_{1,a}(k) e^{2i t g(k)} & 1 \end{pmatrix}, & k \in U_1, 
	\\
\begin{pmatrix} 1 & - r_{2,a}(k) e^{-2it g(k)} \\ 0 & 1 \end{pmatrix}, & k \in U_2, 
	\\
\begin{pmatrix} 1 & 0 \\ -r_{2,a}^*(k)  e^{2it g(k)} & 1 \end{pmatrix}, & k \in U_3, 
	\\
\begin{pmatrix} 1 & r_{1,a}^*(k)e^{-2i t g(k)} \\ 0  & 1 \end{pmatrix}, & k \in U_4,
	\\
\begin{pmatrix} 1 & -h_a^*(k)e^{-2i t g(k)} \\ 0  & 1 \end{pmatrix}, & k \in U_5,
	\\
\begin{pmatrix} 1 & 0 \\ -h_a(k) e^{2i t g(k)} & 1 \end{pmatrix}, & k \in U_6, 
	\\
	I, & \text{elsewhere}.
\end{cases}
\end{align*}
 Then the new jump matrix \(v^{(4)}=(D^{\sigma_3}H_4D^{-\sigma_3})_-^{-1}v^{(3)}(D^{\sigma_3}H_4D^{-\sigma_3})_+\) is given by 
 \begin{align*}
& v_{1}^{(4)} =\begin{pmatrix} 1 & 0 \\ -D^{-2}(r_{1,a}+h) e^{2i t g} & 1 \end{pmatrix},
	\qquad
 v_{2}^{(4)} =\begin{pmatrix} 1 & - D^{2}r_{2,a} e^{-2it g} \\ 0 & 1 \end{pmatrix},
	\\
& v_3^{(4)} =\begin{pmatrix} 1 & 0 \\ D^{-2}r_{2,a}^*  e^{2it g} & 1 \end{pmatrix},
 	\qquad
  v_4^{(4)} =\begin{pmatrix} 1 & D^2(r_{1,a}^*+h^*)e^{-2i t g} \\ 0  & 1 \end{pmatrix},	
  	\\
&   v_5^{(4)} =\begin{pmatrix} 1 & D^2(r_{1,a}^*+h_a^*)e^{-2i t g} \\ 0  & 1 \end{pmatrix},	
	\qquad
 v_6^{(4)} =\begin{pmatrix} 1 & 0 \\ -D^{-2}(r_{1,a}+h_a) e^{2i t g} & 1 \end{pmatrix},
	\\
& v_7^{(4)} 
	=\begin{pmatrix} 1-|r_{1,r}|^2 & D^2r_{1,r}^*e^{-2itg} \\ -D^{-2}r_{1,r} e^{2i t g} & 1 \end{pmatrix}, 	
	\qquad
v_8^{(4)} =\begin{pmatrix} 1 & 0 \\ -D^{-2}h_r e^{2i t g} & 1 \end{pmatrix},
	\\
&v_9^{(4)} =\begin{pmatrix} 1 & D^2h_r^*e^{-2i t g} \\ 0  & 1 \end{pmatrix},
	\qquad
v_{10}^{(4)} 
=\begin{pmatrix}1&D_+^2r_{2,r}e^{-2itg}\\-D_-^{-2}r_{2,r}^*e^{2itg}&1-|r_{2,r}|^2(1-|r|^2)^2\end{pmatrix},
	\\
& v_{11}^{(4)} =\begin{pmatrix}0&1\\-1&D_+^2r_{2,a}e^{-2itg_+}+D_-^{-2}r_{2,a}^*e^{2itg_-}+D_+D_-^{-1}e^{-2itg_+}\end{pmatrix},
\end{align*}
where the subscripts refer to Figure \ref{Gamma2.pdf-new} and we have used Lemma \ref{Dlemma} (e) to compute the jumps along \((-\infty,\kappa_+)\).
Using \(g_+=-g_-\), \(D_+D_-=r\) as well as the assumption \(|r|^2=rr^*=1\) on \((E_1,E_2)\), we find that the \(22\)-entry of \(v_{11}^{(4)}\) can be written as
\begin{align*}
D_+D_-^{-1}(rr_{2,a}+r^*r_{2,a}^*+1)e^{-2itg_+}
	=D_+D_-^{-1}(r r_2+r^*r_2^*+1-2\re(rr_{2,r}))e^{-2itg_+}.
\end{align*}
Using again that \(|r|=1\) on \((E_1,E_2)\), we compute
\begin{align*}
r r_2+r^*r_2^*+1=-\frac{r^2}{1+r^2}-\frac{\bar r^2}{1+\bar r^2}+1=-\frac{|r|^4+r^2+\bar r^2+|r|^4}{1+r^2+\bar r^2+| r|^4}+1=0.
\end{align*}
Thus we can rewrite \(v_{11}^{(4)}\) as
\begin{align*}
v_{11}^{(4)} =\begin{pmatrix}0&1\\-1&-2D_+D_-^{-1}\re(rr_{2,r})e^{-2itg_+}\end{pmatrix}.
\end{align*}

\section{Parametrices}

In this section we will construct parametrices away from and near the critical point. These parametrices will be refered to as the global parametrix and the parametrix near \(k_0\), respectively.

\subsection{Global Parametrix}

The signature table of \(\im g\) (see Figure \ref{signature_table_Img.pdf}) together with the decay of \(r_{1,r}\), \(r_{2,r}\) and \(h_r\) implies that away from the critical point \(k_0\),
the jump matrix $v^{(4)}$  approaches the jump matrix 
\begin{align}\label{vmod}
v^{(\infty)} = \begin{pmatrix} 0 & 1 \\ -1 & 0 \end{pmatrix}, \qquad k \in [E_1, E_2],
\end{align}
as $t \to \infty$ (with the jumps on the other parts of \(\Gamma^{(4)}\) being equal to identity).
Hence we expect the leading order asymptotics of \(m^{(4)}\) to be determined by the solution $m^{(\infty)}$ of the RH problem 
\begin{align*}
\begin{cases}
m^{(\infty)}(x, t, \cdot) \in I + \dot{E}^2(\C \setminus [E_1, E_2]),\\
m^{(\infty)}_+(x,t,k) = m^{(\infty)}_-(x, t, k) v^{(\infty)}(x, t, k) \quad \text{for a.e.} \ k \in [E_1, E_2],
\end{cases}
\end{align*}
where $v^{(\infty)}$ is given by \eqref{vmod} and the contour \([E_1,E_2]\) is oriented to the right.
The unique solution of this RH problem is given by 
\begin{align*}
m^{(\infty)}= \frac{1}{2}\begin{pmatrix} \Delta+ \Delta^{-1} & -i(\Delta - \Delta^{-1}) \\
i(\Delta - \Delta^{-1}) &  \Delta+ \Delta^{-1} \end{pmatrix},
\end{align*}
where the function \(\Delta\) is defined in \eqref{defindelta}.
We note that \(m^{(\infty)}\) satisfies the asymptotic formula
\begin{align*}
m^{(\infty)}(k)=I+\begin{pmatrix}0&\frac{i\alpha}{2k}\\-\frac{i\alpha}{2k}&0\end{pmatrix}+\begin{pmatrix}\frac{\alpha^2}{8k^2}&-\frac{i\alpha\beta}{2k^2}\\\frac{i\alpha\beta}{2k^2}&\frac{\alpha^2}{8k^2}\end{pmatrix}+O(k^{-3}),\quad k\to\infty.
\end{align*}

It follows that \(v^{(4)}-v^{(\infty)}\) converges to zero as \(t\to \infty\) everywhere except at the critical point \(k_0\). Thus we have to do a local analysis near \(k_0\).

\subsection{Model Problem on the Cross}

The study of the local parametrix near \(k_0\) leads to a RH problem on a cross which can be explicitly solved in terms of parabolic cylinder functions \cite{Its1981}.
The exact result needed in our case can be found in Appendix B of \cite{NonLinSteepDescent-Lenells} and is stated below for the reader's convenience.

Let \(X=X_1\cup X_2\cup X_3\cup X_4\subseteq \C\) be the cross defined by
\begin{align*}
X_1&=\big\{se^{\frac{i\pi}4}\,\big|\,0\leq s<\infty\big\}, &X_2&=\big\{se^{\frac{3i\pi}4}\,\big|\,0\leq s<\infty\big\},
\\
X_3&=\big\{se^{-\frac{3i\pi}4}\,\big|\,0\leq s<\infty\big\},&X_4&=\big\{se^{-\frac{i\pi}4}\,\big|\,0\leq s<\infty\big\},
\end{align*}
oriented away from the origin as shown in Figure \ref{crossexactsolution.pdf} and let \(\mathbb{D}=\{z\in \C\,|\,|z|<1\}\) denote the open unit disk in \(\C\).

\begin{figure}
\begin{center}
 \begin{overpic}[width=.4\textwidth]{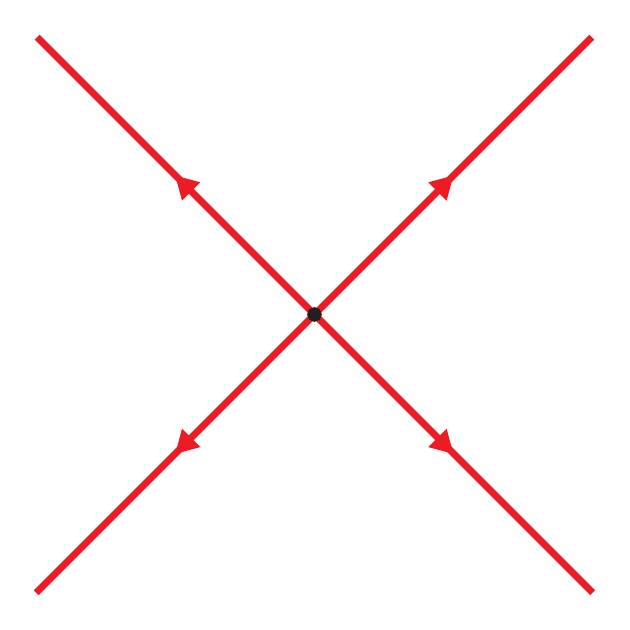}
      \put(19,67){\small $X_2$}
      \put(19,31){\small $X_3$}
      \put(73,67){\small $X_1$}
      \put(73,31){\small $X_4$}
      \put(48.5,42){\small $0$}
    \end{overpic}
     \begin{figuretext}\label{crossexactsolution.pdf}
        The contour  $X=X_1\cup X_2\cup X_3\cup X_4$ in the complex plane.
         \end{figuretext}
     \end{center}
\end{figure}

\begin{lemma}[Exact solution on the cross]\label{lemmaexactsolution}
Define the function $\nu:\mathbb{D} \to [0,\infty)$ by 
$\nu(q) = -\frac{1}{2\pi} \ln(1 - |q|^2)$ and define the jump matrix $v^X(q, z)$ for $z \in X$ by
\begin{align}\label{vXdef} 
v^X(q, z) = \begin{cases}
\begin{pmatrix} 1 & 0	\\
  -q z^{-2i\nu(q)} e^{\frac{iz^2}{2}}	& 1 \end{pmatrix}, &   z \in X_1, 
  	\\
\begin{pmatrix} 1 & -\frac{\bar{q}}{1 - |q|^2} z^{2i\nu(q)}e^{-\frac{iz^2}{2}}	\\
0 & 1  \end{pmatrix}, &  z \in X_2, 
	\\
\begin{pmatrix} 1 &0 \\
 \frac{q}{1 - |q|^2}z^{-2i\nu(q)} e^{\frac{iz^2}{2}}	& 1 \end{pmatrix}, &  z \in X_3,
	\\
 \begin{pmatrix} 1	& \bar{q} z^{2i\nu(q)}e^{-\frac{iz^2}{2}}	\\
0	& 1 \end{pmatrix}, &  z \in X_4.
\end{cases}
\end{align}
Then for each $q \in \mathbb{D}$ the RH problem 
\begin{align*}
\begin{cases} m^X(q, \cdot) \in I + \dot{E}^2(\C \setminus X), 
	\\
m_+^X(q, z) =  m_-^X(q, z) v^X(q, z) \quad \text{for a.e.} \ z \in X, 
\end{cases} 
\end{align*}
has a unique solution $m^X(q, z)$. This solution satisfies
\begin{align*}
  m^X(q, z) = I - \frac{1}{z}\begin{pmatrix} 0 & \beta^X(q) \\ \overline{\beta^X(q)} & 0 \end{pmatrix} + O\biggl(\frac{q}{z^2}\biggr), \qquad z \to \infty,  \ q \in \mathbb{D}, 
\end{align*}  
where the error term is uniform with respect to $\arg z \in [0, 2\pi]$ and $q$ in compact subsets of $\mathbb{D}$. The function $\beta^X(q)$ is defined by
\begin{align}\label{betaXdef}
\beta^X(q) = \sqrt{\nu(q)} e^{i\left(\frac{3\pi}{4} - \arg (-q) + \arg \Gamma(i\nu(q)\right)}, \qquad q \in \mathbb{D}.
\end{align}
Moreover, for each compact subset $K$ of $\mathbb{D}$, 
\begin{align}\label{mXbound}
\sup_{q \in K} \sup_{z \in \C \setminus X} |m^X(q, z)| < \infty
\end{align}
and
\begin{align}\label{mXqbound}
\sup_{q \in K} \sup_{z \in \C \setminus X} \frac{|m^X(q, z)- I|}{|q|} < \infty.
\end{align}
\end{lemma}

\begin{proof} See Theorem B.1 in \cite{NonLinSteepDescent-Lenells} with \(q\) replaced by \(-q\). 
\end{proof}

\subsection{Parametrix near $k_0$}

Fix \(0<\epsilon<(\inf_{\zeta\in[0,c_0]} k_0(\zeta)-E_2)/2\). 
We seek a \(2\times2\) matrix valued function \(m^{k_0}\) with jumps along \(\Gamma^{(4)}\cap D_\epsilon(k_0)\) such that the corresponding jump matrix is close to \(v^{(4)}\) and such that \(m^{k_0}\) is close to \(m^{(\infty)}\)  on \(\partial D_\epsilon(k_0)\) for large \(t\). 
In order to find \(m^{k_0}\), we relate \(m^{(4)}\) to the solution \(m^X\) of Lemma \ref{lemmaexactsolution} by making a local change of variables for \(k\) near \(k_0\).

We introduce a new variable \(z\equiv z(\zeta,k)\) such that
\begin{align*}
\frac{iz^2}2=2it(g(k)-g(k_0)).
\end{align*}
Hence we choose
\begin{align*}
z=\sqrt{t}(k-k_0)\psi(\zeta,k)
\end{align*}
where 
\begin{align}\label{defpsi}
\psi(\zeta,k)=2\sqrt{\frac{g(k)-g(k_0)}{(k-k_0)^2}}.
\end{align}
Note that \(\partial_kg(k_0)=0\) and \(\partial_k^2g(k_0)=4(k_0-k_1)/{X(k_0)}> 0\) (cf. \eqref{gfunction-derivative}), so that \(\psi^2\) is analytic and non-zero in a neighbourhood of \(k_0\). Since \(\psi^2\) is continuous with respect to \((\zeta,k)\), we may (by making \(\epsilon\) smaller) assume that \(\psi^2\) is non-zero in \(D_\epsilon(k_0)\).
We fix the branch of the square root in \eqref{defpsi}  by requiring that \(\re \psi(\zeta,k)>0\) for \(k \in D_\epsilon(k_0)\).

By making \(\epsilon\) smaller if necessary, we may assume that that for each \(\zeta\in[0,c_0]\), the map \(k\mapsto z(\zeta,k)\) is a biholomorphism from \(D_\epsilon(k_0)\) onto some neighbourhood of the origin.
This follows from the fact that \(\sqrt{\partial_k^2g(k_0)}\)  can be uniformly bounded from below and above for \(\zeta\in [0,c_0]\).

Let \(\mathcal{X}\) be the cross defined by parts \(1\), \(2\), \(3\), \(4\), \(5\) and \(6\) of the contour \(\Gamma^{(4)}\) and let \(\mathcal{X}^\epsilon=\mathcal{X}\cap D_\epsilon(k_0)\).
By deforming the contour slightly we may assume that \(\mathcal{X}^\epsilon\) is mapped into \(X\) under the map \(k\mapsto z(\zeta,k)\). 
Due to the symmetry \(z(\zeta,k)=\overline{z(\zeta,\overline{k})}\) we may in addition assume that the deformed contour is invariant (up to orientation) under the involution  \(k\mapsto \bar k\).

Using Lemma \ref{Dlemma} (g) we write the function \(D\) for \(k\in D_\epsilon(k_0)\) as 
\begin{align*}
D(\zeta,k)
&=(k-k_0)^{{i}\nu}D_b(\zeta,k)
=z^{i\nu}D_0(\zeta,t)D_1(\zeta,k)
\end{align*} 
where the functions \(D_0(\zeta,t)\) and \(D_1(\zeta,k)\) are defined by
\begin{align*}
D_0(\zeta,t)&=t^{\frac{-i\nu}2}\psi(\zeta,k_0)^{-i\nu}D_b(\zeta,k_0),\quad t>0,
\end{align*}
and
\begin{align*}
D_1(\zeta,k)&=e^{-i\nu\ln\left(\frac{\psi(\zeta,k)}{\psi(\zeta,k_0)}\right)}\frac{D_b(\zeta,k)}{D_b(\zeta,k_0)},\quad k\in D_\epsilon(k_0)\setminus (-\infty,k_0].
\end{align*}
Define \(\tilde m(x,t,z)\) by
\begin{align*}
\tilde m(x,t,z(\zeta,k))=m^{(4)}(x,t,k)e^{-itg(\zeta,k_0)\sigma_3}D_0(\zeta,t)^{\sigma_3},\quad k\in D_\epsilon(k_0)\setminus \Gamma^{(4)}.
\end{align*}
Then \(\tilde m\) is a sectionally analytic function which satisfies \(\tilde m_+=\tilde m_-\tilde v\) for \(k\in \Gamma^{(4)}\cap D_\epsilon(k_0)\), with the jump matrix
\begin{align*}
\tilde v=D_0(\zeta,t)^{-\sigma_3}e^{itg(\zeta,k_0)\sigma_3}v^{(4)}e^{-itg(\zeta,k_0)\sigma_3}D_0(\zeta,t)^{\sigma_3}
\end{align*}
 given by
\begin{align*}
\tilde v(x,t,z)
&=\begin{cases}
 \begin{pmatrix} 1 & 0 \\ -z^{-2i\nu}D_1^{-2}(r_{1,a}+h) e^{\frac{iz^2}2} & 1 \end{pmatrix},&k\in\Gamma_1^{(4)}\cap D_\epsilon(k_0),
	\\
\begin{pmatrix} 1 & - z^{2i\nu}D_1^2r_{2,a} e^{-\frac{iz^2}2} \\ 0 & 1 \end{pmatrix},&k\in\Gamma_2^{(4)}\cap D_\epsilon(k_0),
	\\
 \begin{pmatrix} 1 & 0 \\ z^{-2i\nu}D_1^{-2}r_{2,a}^*  e^{\frac{iz^2}2} & 1 \end{pmatrix},&k\in\Gamma_3^{(4)}\cap D_\epsilon(k_0),
	\\
 \begin{pmatrix} 1 & z^{2i\nu}D_1^2(r_{1,a}^*+h^*)e^{-\frac{iz^2}2} \\ 0  & 1 \end{pmatrix},&k\in\Gamma_4^{(4)}\cap D_\epsilon(k_0),
 	\\
 \begin{pmatrix} 1 & z^{2i\nu}D_1^2(r_{1,a}^*+h_a^*)e^{-\frac{iz^2}2} \\ 0  & 1 \end{pmatrix},&k\in\Gamma_5^{(4)}\cap D_\epsilon(k_0),
	\\
 \begin{pmatrix} 1 & 0 \\ -z^{-2i\nu}D_1^{-2}(r_{1,a}+h_a) e^{\frac{iz^2}2} & 1 \end{pmatrix},&k\in\Gamma_6^{(4)}\cap D_\epsilon(k_0),
	\\
 \begin{pmatrix} 1-|r_{1,r}|^2 & z^{2i\nu}D_1^2r_{1,r}^*e^{-\frac{iz^2}2} \\ -z^{-2i\nu}D_1^{-2}r_{1,r} e^{\frac{iz^2}2} & 1 \end{pmatrix},&k\in\Gamma_{7}^{(4)}\cap D_\epsilon(k_0),
 	\\
\begin{pmatrix} 1 & 0 \\ -z^{-2i\nu}D_1^{-2}h_{r} e^{\frac{iz^2}2} & 1 \end{pmatrix},&k\in\Gamma_8^{(4)}\cap D_\epsilon(k_0),
 	 \\
\begin{pmatrix} 1 & z^{2i\nu}D_1^2h_{r}^*e^{-\frac{iz^2}2} \\ 0  & 1 \end{pmatrix},&k\in\Gamma_9^{(4)}\cap D_\epsilon(k_0),
	\\
\begin{pmatrix}1&(z^{2i\nu})_+D_{1+}^2r_{2,r}e^{-\frac{iz^2}2}\\-(z^{-2i\nu})_-D_{1-}^{-2}r_{2,r}^*e^{\frac{iz^2}2}&1-|r_{2,r}|^2(1-|r|^2)^2\end{pmatrix},&k\in\Gamma_{10}^{(4)}\cap D_\epsilon(k_0).
\end{cases}
\end{align*}
Note that the parts \(\Gamma_{i}^{(4)}\cap D_\epsilon(k_0)\) for \(i=5,6,8,9\), are non-empty only for \(k_0\) close to \(\kappa_+\).
Define
\begin{align*}
q\equiv q(\zeta)=r(k_0),\quad \zeta\in[0,c_0].
\end{align*}
For fixed \(z\in X\) we have \(k(\zeta,z)\to k_0\) as \(t\to \infty\),
which implies
\begin{align*}
D_1(\zeta,k)\to 1, \quad r_{1,a}(k)+h(k)\to r(k_0)= q,\quad r_{2,a}(k)\to \frac{\overline{r(k_0)}}{1-|r(k_0)|^2}=\frac{\bar q}{1-|q|^2},
\end{align*}
as \(t\to \infty\). 
Hence we expect \(\tilde v\) to tend to the jump matrix \(v^X\) defined in \eqref{vXdef} as \(t\) becomes large. By definition of \(\tilde m\) this means that the jumps of \(m^{(4)}\) for \(k\) near \(k_0\) approach the jumps of the function \(m^XD_0(\zeta,t)^{-\sigma_3}e^{itg(\zeta,k_0)\sigma_3}\) as \(t\to \infty\). 
A suitable approximation of \(m^{(4)}\) in the neighbourhood \(D_\epsilon(k_0)\) of \(k_0\) is thus given by a \(2\times 2\)-matrix valued function \(m^{k_0}\) of the form
\begin{align}\label{approximationnearcross}
m^{k_0}(x,t,k)=Y(\zeta,t,k)m^X(q(\zeta),z(\zeta,k))D_0(\zeta,t)^{-\sigma_3}e^{itg(\zeta,k_0)\sigma_3}
\end{align}
where \(Y(\zeta,t,k)\) is a matrix valued function which is analytic for \(k\in D_\epsilon(k_0)\). 
Since we want \(m^{k_0}\) to be close to \(m^{(\infty)}\) on \(\partial D_\epsilon(k_0)\) for large \(t\), we choose 
\begin{align}\label{defY}
Y(\zeta,t,k)=m^{(\infty)}(k)e^{-itg(\zeta,k_0)\sigma_3}D_0(\zeta,t)^{\sigma_3}.
\end{align}

\begin{lemma}\label{localmodellemma}
For each \(\zeta\in [0,c_0]\) and \(t>0\), the function \(m^{k_0}(x,t,k)\) defined in \eqref{approximationnearcross} with \(Y\) given by \eqref{defY} is an analytic function of \(k\in D_\epsilon(k_0)\setminus \mathcal{X}^\epsilon\). Moreover,
\begin{align*}
|(m^{(\infty)})^{-1}m^{k_0}(x,t,k)-I|\leq C|q|\leq C,\quad \zeta\in[0,c_0],t>2,k\in \overline{D_\epsilon(k_0)}\setminus \mathcal{X}^\epsilon.
\end{align*}
Across \(\mathcal{X}^\epsilon\), \(m^{k_0}\) obeys the jump condition \(m_+^{k_0}=m_-^{k_0}v^{k_0}\), where the jump matrix \(v^{k_0}\) satisfies
\begin{align}\label{Lpestimatesjufmp}
\begin{cases}
\|v^{(4)}-v^{k_0}\|_{L^1(\mathcal{X}^\epsilon)}\leq Ct^{-1}\ln t,
\\
\|v^{(4)}-v^{k_0}\|_{L^2(\mathcal{X}^\epsilon)}\leq Ct^{-3/4}\ln t,
\\
\|v^{(4)}-v^{k_0}\|_{L^\infty(\mathcal{X}^\epsilon)}\leq Ct^{-1/2}\ln t,
\\
\|v^{(4)}-v^{k_0}\|_{(L^1\cap L^2\cap L^\infty)(( \Gamma^{(4)}_{7}\cup \Gamma^{(4)}_{8}\cup \Gamma^{(4)}_{9}\cup\Gamma^{(4)}_{10})\cap D_\epsilon(k_0))}\leq Ct^{-3/2},
\end{cases}
\end{align}
uniformly for \(\zeta\in[0,c_0]\) and \(t>2\).
Furthermore, as \(t\to\infty\),
\begin{align}\label{asymptoticsmk0-1-Ipart1}
\|m^{(\infty)}(m^{k_0}(x,t,\cdot))^{-1}-I\|_{L^\infty(\partial D_\epsilon(k_0))}=O\left(\frac{q}{\sqrt{t}}\right),
\end{align}
and, for \(j=0,1\), 
\begin{align}\label{asymptoticsmk0-1-Ipart2mmod}
\frac{1}{2\pi i}&\int_{\partial D_\epsilon(k_0)}k^j\left(m^{(\infty)}(k)(m^{k_0}(x,t,k))^{-1}-I\right){dk}\nonumber
	\\&=-k_0^j\frac{Y(\zeta,t,k_0)m_1^X(\zeta)Y(\zeta,t,k_0)^{-1}}{\sqrt{t}\psi(\zeta,k_0)}+O\left(\frac{q}{t}\right),
\end{align}
uniformly with respect to \(\zeta\in[0,c_0]\), where \(m_1^X(\zeta)\) is defined by
\begin{align}\label{definitionm1X}
m_1^X(\zeta)=-\begin{pmatrix} 0 & \beta^X(q) \\ \overline{\beta^X(q)} & 0 \end{pmatrix}.
\end{align}
\end{lemma}

\begin{proof}
See for instance Lemma 6.3 in \cite{DNLSArrudaLenells}.
 The necessary estimates for the analytic approximations of \( r_2\), \(r_1\) and \(h\), as well as for the functions \(D\), \(D_0\) and \(D_1\), follow from Lemma \ref{analyticapproximation r2}, \ref{analyticapproximationr1} and \ref{analyticapproximationh}, and Lemma \ref{Dlemma} together with the uniform boundedness of \(\nu\), respectively.
\end{proof}

\section{Asymptotic Analysis}\label{asymptoticanalysis}

Define the approximate solution \(m^{app}\) by
\begin{align*}
m^{app}=\begin{cases}m^{k_0},& k\in D_\epsilon(k_0),
	\\
	m^{(\infty)},& \textrm{else}.\end{cases}
\end{align*}
The function \(\hat m(x,t,k)\) defined by
\begin{align}\label{definitionmhat}
\hat m=m^{(4)}(m^{app})^{-1}
\end{align}
satisfies the jump relation
\begin{align*}
\hat m_+(x,t,k) = \hat m_-(x, t, k) \hat v(x, t, k) \quad \text{for a.e.} \ k \in \hat \Gamma,
\end{align*}
where \(\hat \Gamma=\Gamma^{(4)}\cup\partial D_\epsilon(k_0)\), see Figure \ref{Gammahat.pdf-new}, and the jump matrix \(\hat v\) is given by
\begin{align*}
\hat v=\hat m_-^{-1}\hat m_+=\begin{cases}
	m_-^{(\infty)}v^{(4)}(m_+^{(\infty)})^{-1},& k\in \hat\Gamma\setminus\overline{D_{\epsilon}(k_0)},
		\\
	m^{(\infty)}(m^{k_0})^{-1},&k \in\partial D_\epsilon(k_0),
		\\
	m_-^{k_0}v^{(4)}(m_+^{k_0})^{-1},& k\in \hat\Gamma\cap{D_{\epsilon}(k_0)}.
	\end{cases}
\end{align*}
The next lemma shows that \(\hat v-I\) is small for  large \(t\).
\begin{figure}
\begin{center}
 \begin{overpic}[width=.7\textwidth]{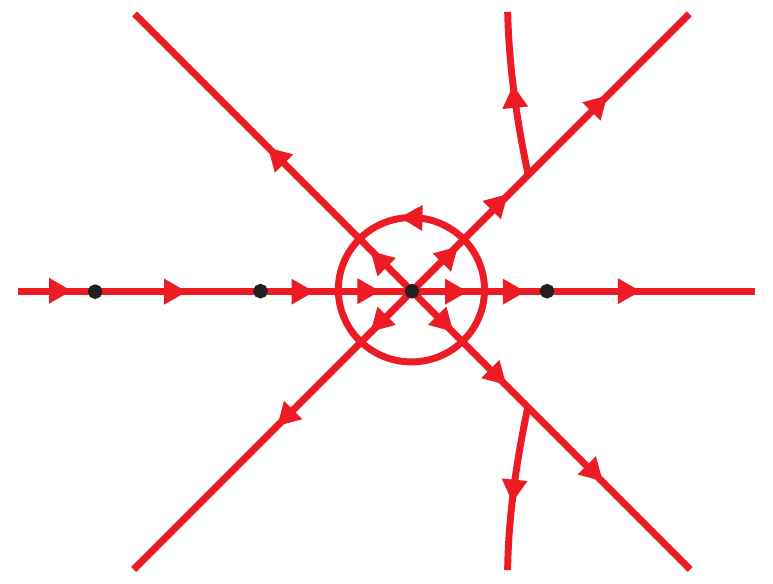}
      \put(100,36.5){\small $\hat\Gamma$}
      \put(10.5,33){\small $E_1$}
      \put(32,33){\small $E_2$}
      \put(70,34){\small $\kappa_+$}
      \put(51.9,33){\small $k_0$}
    \end{overpic}
     \begin{figuretext}\label{Gammahat.pdf-new}
        The contour  $\hat\Gamma$ in the complex $k$-plane. 
         \end{figuretext}
     \end{center}
\end{figure}

\begin{lemma}\label{omegahatestimateslemma}
Let \(\hat w=\hat v-I\). Then the following estimates hold uniformly for \(\zeta\in[0,c_0]\) and \(t\geq2\): 
\begin{subequations}\label{omegahatestimates}
\begin{align}
&\|(1+|k|^2)\hat w\|_{(L^1\cap L^2\cap L^\infty)(\mathcal{X}\setminus\overline{D_\epsilon(k_0)})}\leq Ce^{-ct},\label{omegahatestimatecrossoutsidedisk}
	\\
&\|(1+|k|^2)\hat w\|_{(L^1\cap L^2\cap L^\infty)(\Gamma^{(4)}\setminus(\mathcal{X}\cup\partial{D_\epsilon(k_0)}))}\leq Ct^{-3/2},\label{omegahatestimatereallineoutsidedisk}
	\\
&\|(1+|k|^2)\hat w\|_{(L^1\cap L^2\cap L^\infty)(\partial D_\epsilon(k_0))}\leq Cqt^{-1/2},\label{omegahatestimatereallineboundarydisk}
	\\
&\|(1+|k|^2)\hat w\|_{L^1(\mathcal{X}^\epsilon)}\leq Ct^{-1}\ln t,\label{omegahatestimaterecross1}
	\\
&\|(1+|k|^2)\hat w\|_{L^2(\mathcal{X}^\epsilon)}\leq Ct^{-3/4}\ln t,\label{omegahatestimaterecross2}
	\\
&\|(1+|k|^2)\hat w \|_{L^\infty(\mathcal{X}^\epsilon)}\leq Ct^{-1/2}\ln t.\label{omegahatestimaterecross3}
\end{align}
\end{subequations}
\end{lemma}

\begin{proof}
Note that
\begin{align*}
\hat w=\begin{cases}
	m^{(\infty)}(v^{(4)}-I)(m^{(\infty)})^{-1},& k\in \hat\Gamma\setminus(\overline{D_{\epsilon}(k_0)}\cup[E_1,E_2]),
		\\
	m_-^{(\infty)}(v^{(4)}-v^{(\infty)})(m_+^{(\infty)})^{-1},& k\in [E_1,E_2],
		\\
	m^{(\infty)}(m^{k_0})^{-1}-I,&k \in\partial D_\epsilon(k_0),
		\\
	m_-^{k_0}(v^{(4)}-v^{k_0})(m_+^{k_0})^{-1},& k\in \hat\Gamma\cap{D_{\epsilon}(k_0)}.
	\end{cases}
\end{align*}
Thus the estimates \eqref{omegahatestimaterecross1}-\eqref{omegahatestimaterecross3} are a consequence of \eqref{Lpestimatesjufmp} and \eqref{asymptoticsmk0-1-Ipart1} immediately gives \eqref{omegahatestimatereallineboundarydisk}.

For the estimate on part \(2\) of the contour lying outside the disk we note that the only nonzero entry of \(v^{(4)}-I\) is given by
\(
-D^2(k)r_{2,a}(k)e^{-2itg(k)}.
\) 
Since \(D\) is uniformly bounded with respect to \(k\) and \(\zeta\in[0,c_0]\) (cf. Lemma \ref{Dlemma} (h)), the estimate \eqref{ r2aestimatesecond} implies that the above term can be bounded by a uniform constant times \(e^{-\frac74|\im g(k)|}\). 
However, on \(\mathcal{X}\setminus\overline{D_\epsilon(k_0)}\) we have \(|\im g(k)|\geq c |k-k_0|\) for some uniform constant \(c>0\). 
This yields the estimate \eqref{omegahatestimatecrossoutsidedisk} on part \(2\) of the cross \(\mathcal{X}\setminus\overline{D_\epsilon(k_0)}\).
The other parts of the cross can be treated similarly.

Next we derive the estimate \eqref{omegahatestimatereallineoutsidedisk} on part \(11\) of the contour.
Using \(\Delta_+=i\Delta_-\) we compute
\begin{align*}
\hat w
	&=m_-^{(\infty)}\begin{pmatrix}0&0\\0&-2D_+D_-^{-1}\re(rr_{2,r})e^{-2itg_+}\end{pmatrix}(m_+^{(\infty)})^{-1}
	\\
	&=-\frac{\re(rr_{2,r})e^{-2itg_+}}2\begin{pmatrix}-iD_+D_-^{-1}(\Delta_-^2-\Delta_-^{-2})&-D_+D_-^{-1}(\Delta_-^2+\Delta_-^{-2}-2)\\D_+D_-^{-1}(\Delta_-^2+\Delta_-^{-2}+2)&iD_+D_-^{-1}(\Delta_-^2-\Delta_-^{-2})\end{pmatrix}.
\end{align*}
Lemma \ref{Dlemma} f) implies that the components of the last matrix are uniformly bounded on \([E_1,E_2]\) with respect to \(\zeta\in[0,c_0]\). 
In view of Lemma \ref{analyticapproximation r2} (c), the estimate \eqref{omegahatestimatereallineoutsidedisk} now follows for part \(11\) of the contour. The remaining parts can be treated in the same way since the jump matrix contains the small remainders \(r_{1,r}\), \(r_{2,r}\) or \(h_r\) so that the desired estimates are a consequence of Lemma \ref{analyticapproximationr1} (c), \ref{analyticapproximation r2} (c) and Lemma \ref{analyticapproximationh} (c).
\end{proof}

The estimates in Lemma \ref{omegahatestimateslemma} imply that
\begin{align}\label{omegahatcollectedestimates}
\begin{cases}
\|\hat w\|_{(L^1\cap L^2)(\hat \Gamma)}\leq Ct^{-1/2},
	\\
\|\hat w\|_{L^\infty(\hat \Gamma)}\leq Ct^{-1/2}\ln t,
\end{cases}
\qquad \zeta\in[0,c_0],\quad t\geq 2.
\end{align}
For \(f\in L^2(\hat \Gamma)\) the Cauchy transform  \( {\mathcal{C}} f\) is defined by
\begin{align}\label{Cauchytransform}
({\mathcal{C}} f)(\lambda)=\frac{1}{2\pi i}\int_{\hat \Gamma}\frac{f(z)}{z-\lambda} dz,\quad \lambda\in \C\setminus \hat \Gamma.
\end{align}
Let \(\mathcal{C}_+f\) and \(\mathcal{C}_-f\) denote the nontangential boundary values of \(\mathcal{C}f\) from the left and right sides of \(\hat \Gamma\), respectively. Then \(\mathcal{C}_+\) and \(\mathcal{C}_-\) lie in \(\mathcal{B}(L^2(\hat\Gamma))\) and satisfy \(\mathcal{C}_+-\mathcal{C}_-=I\), where \(\mathcal{B}(L^2(\hat\Gamma))\) denotes the Banach space of bounded linear operators on \(L^2(\hat\Gamma)\). 

The estimates \eqref{omegahatcollectedestimates} imply
\begin{align}\label{normComegaoperator}
\|\mathcal{C}_{\hat w}\|_{\mathcal{B}(L^2(\hat \Gamma))}\leq C\|\hat w\|_{L^\infty(\hat \Gamma)}\leq Ct^{-1/2}\ln t,\quad\zeta\in[0,c_0],\ t\geq 2,
\end{align}
with the operator \(\mathcal{C}_{\hat w}\colon L^2(\hat \Gamma)+L^\infty(\hat \Gamma)\to L^2(\Hat \Gamma)\) being defined by \(\mathcal{C}_{\hat w}f=\mathcal{C}_-(f\hat w)\).
It follows that there exists a time \(T>0\)  such that \(\|C_{\hat w}\|_{\mathcal{B}(L^2(\hat\Gamma))}\leq \frac12\) and \(I-\mathcal{C}_{\hat w(\zeta,t,\cdot)}\in \mathcal{B}(L^2(\hat \Gamma))\) is invertible  for all \(t\geq T\). 
We define \(\hat \mu(x,t,k)\in I+L^2(\hat \Gamma)\) for \(t \geq T\) by
\begin{align*}
\hat\mu=I+(I-\mathcal{C}_{\hat w})^{-1}\mathcal{C}_{\hat w}I.
\end{align*}
Using the Neumann series, we find that
\begin{align*}
\|\hat\mu-I\|_{L^2(\hat \Gamma)}\leq \frac{C\|\hat w\|_{L^2(\hat\Gamma)}}{1-\|\mathcal{C}_{\hat w}\|_{\mathcal{B}(L^2(\hat\Gamma))}},\quad t\geq T.
\end{align*}
Together with \eqref{omegahatcollectedestimates} and \eqref{normComegaoperator} it  follows that
\begin{align}\label{muhat-Iestimate}
\|\hat\mu(x,t,\cdot)-I\|_{L^2(\hat\Gamma)}\leq Ct^{-1/2},\quad t\geq T,\zeta\in [0,c_0].
\end{align}
The standard theory of \(L^2\)-RH problems now implies that there exists a unique solution \(\hat m\in I+\dot{E}^2( \C\setminus\hat\Gamma)\) of the RH problem 
\begin{align}\label{RHmhat}
\begin{cases}
\hat m(x, t, \cdot) \in I + \dot{E}^2(\C \setminus \hat\Gamma),\\
\hat m_+(x,t,k) = \hat m_-(x, t, k) \hat v(x, t, k) \quad \text{for a.e.} \ k \in \hat \Gamma,
\end{cases}
\end{align}
 for all \(t\geq T\). 
This solution is given by
\begin{align}\label{solutionmhat}
\hat m(x,t,k)=I+\mathcal{C}(\hat\mu\hat w)=I+\frac{1}{2\pi i}\int_{\hat \Gamma}\hat \mu(x,t,z)\hat w(x,t,z)\frac{dz}{z-k}.
\end{align}
For more details we refer for instance to \cite{CarlesonContoursLenells}.

\subsection{Asymptotics of \(\hat m\)} 

Let \(W\) be a nontangential sector at \(\infty\) with respect to \(\hat\Gamma\).
By \eqref{solutionmhat} we may write
\begin{align*}
\hat m(x,t,k)=I-\frac{1}{2\pi i}\int_{\hat\Gamma}(\hat\mu\hat w)(x,t,z)\left(\frac1k+\frac{z}{k^2}+\frac{z^2}{k^3}+\frac{z^3}{k^3(k-z)}\right)dz.
\end{align*}
Note that the quotient \(z/(z-k)\) can be bounded uniformly for  \(z\in\hat\Gamma\) and \(k\in W\) large enough. Furthermore, the \(L^2\)-norm of \(\hat\mu(x,t,\cdot)-I\) is bounded according to \eqref{muhat-Iestimate} and the \(L^1\) and \(L^2\)-norms of \(z^2 \hat w(x,t,z)\) are bounded due to Lemma \ref{omegahatestimateslemma}.
Thus we find
\begin{align*}
\hat m(x,t,k)=I+\frac{\hat m_1(x,t)}k+\frac{\hat m_2(x,t)}{k^2}+O(k^{-3}),\qquad k\in W,
\end{align*}
where the error term is uniform with respect to \(k\in W\) and the coefficients \(\hat m_j\) are given by
\begin{align}\label{mhatasymptotics}
\hat m_j(x,t)&:=-\frac{1}{2\pi i}\int_{\hat \Gamma}\hat\mu(x,t,k)\hat w(x,t,k) k^{j-1}dk,\quad j=1,2.
\end{align}

Next we will compute the asymptotics of \(\hat m_1\) and \(\hat m_2\) as \(t\to\infty\).
By \eqref{omegahatestimates} and \eqref{muhat-Iestimate}, we have
\begin{align*}
\int_{\Gamma'}&\hat\mu(x,t,k)\hat w(x,t,k)k^{j-1}dk
\\
	&=\int_{\Gamma'}\hat w(x,t,k)k^{j-1}dk+\int_{\Gamma'}(\hat\mu(x,t,k)-I)\hat w(x,t,k)k^{j-1}dk
	\\
	&=O(\|k^{j-1}\hat w\|_{L^1(\Gamma')})+O(\|\hat\mu-I\|_{L^2(\Gamma')}\|k^{j-1}\hat w\|_{L^2(\Gamma')})
	\\
	&=O(t^{-3/2}),\quad t\to\infty,\quad j=1,2,
\end{align*}
uniformly in \(\zeta\in[0,c_0]\), where \(\Gamma'\!:=\Gamma^{(4)}\setminus(\mathcal{X}^\epsilon\cup \partial D_\epsilon(k_0))\). 
It follows that the contribution to the integrals in \eqref{mhatasymptotics}  from \(\Gamma'\) are \(O(t^{-3/2})\).
Similarly, the estimates \eqref{omegahatestimates} and \eqref{muhat-Iestimate} show that the contribution from \(\mathcal{X}^\epsilon\) to the right-hand side of \eqref{mhatasymptotics} is \(O(t^{-1}\ln t)\) as \(t\to\infty\) for \(j=1,2\).
By \eqref{asymptoticsmk0-1-Ipart2mmod}, \eqref{omegahatestimatereallineboundarydisk}, and \eqref{muhat-Iestimate},  the contribution from \(\partial D_\epsilon(k_0)\) to the right-hand side of \eqref{mhatasymptotics} is given by
\begin{align*}
-\frac{1}{2\pi i}&\int_{\partial D_\epsilon(k_0)}\hat w(x,t,k)k^{j-1}dk-\frac{1}{2\pi i}\int_{\partial D_\epsilon(k_0)}(\hat\mu(x,t,k)-I)\hat w(x,t,k)k^{j-1}dk
	\\
&=-\frac{1}{2\pi i}\int_{\partial D_\epsilon(k_0)}(m^{(\infty)}(k)(m^{k_0}(k))^{-1}-I)k^{j-1}dk
	\\&\hspace{10em}+O(\|\hat\mu-I\|_{L^2(\partial D_\epsilon(k_0))}\|k^{j-1}\hat w\|_{L^2(\partial D_\epsilon(k_0))})
	\\ 
&=k_0^{j-1}\frac{Y(\zeta,t,k_0)m_1^X(\zeta)Y(\zeta,t,k_0)^{-1}}{\sqrt{t}\psi(\zeta,k_0)}+O\left(\frac{q}t\right),\quad t\to\infty,j=1,2.
\end{align*}
Collecting the above contributions, it follows that
\begin{align}\label{mhatasymptoticspart2}
\hat m_{j}(x,t)
=k_0^{j-1}\frac{Y(\zeta,t,k_0)m_1^X(\zeta)Y(\zeta,t,k_0)^{-1}}{\sqrt{t}\psi(\zeta,k_0)}+O\left(\frac{\ln t}{t}\right),\quad j=1,2,
\end{align}
as \(t\to\infty\), uniformly with respect to \(\zeta\in[0,c_0]\). 

\section{Proof of the Main Result}

Using the results from the Sections \ref{sectiontransformationRHP}-\ref{asymptoticanalysis}, we are now ready to prove Theorem \ref{maintheorem}.

\subsection{The Solution of the Associated RH-problem}\label{rigoroustransformations}

After having solved the small norm RH problem \eqref{RHmhat}, our next step is to show that we can revert the transformations in Section \ref{sectiontransformationRHP} to obtain a solution of our original RH problem \eqref{RHm}.
However, since some of the transformation matrices are singular at the branch points, it is not clear that the corresponding RH problems are equivalent in the setting of Smirnoff classes.
This technical detail can be overcome by looking at the combined transformation matrix, which as shown below is bounded near the branch points.

Taking into account the transformations of Section \ref{sectiontransformationRHP}  as well as \eqref{definitionmhat}, we find that for \(\zeta\in[0,c_0]\) and \(t\geq T\) a solution \(m\) of the RH problem \eqref{RHm} can be formally constructed as
\begin{align}\label{fulltransformation}
m
	&=e^{itg_\infty\sigma_3}D(\zeta,\infty)^{-\sigma_3}\hat m m^{app}D(\zeta,k)^{\sigma_3}H_4^{-1}H_2^{-1}e^{-i(t g(k) - kx - 2k^2 t)\sigma_3},
\end{align}
where \(\hat m\in I+\dot{E}^2( \C\setminus\hat\Gamma)\) given by \eqref{solutionmhat} is the solution of the RH problem \eqref{RHmhat}.
Thus in order to show the existence of the RH problem \eqref{RHm}, it suffices to show that the right hand side of \eqref{fulltransformation} solves the original RH problem \eqref{RHm} in the setting of Smirnoff classes.
The necessary theory is standard, see for instance \cite{CarlesonContoursLenells}, in particular Theorem 5.12 therein.
Since the result needed in our case differs slightly from the one given in the literature we state it below for the reader's convenience.

\begin{figure}
\begin{center}
 \begin{overpic}[width=.4\textwidth]{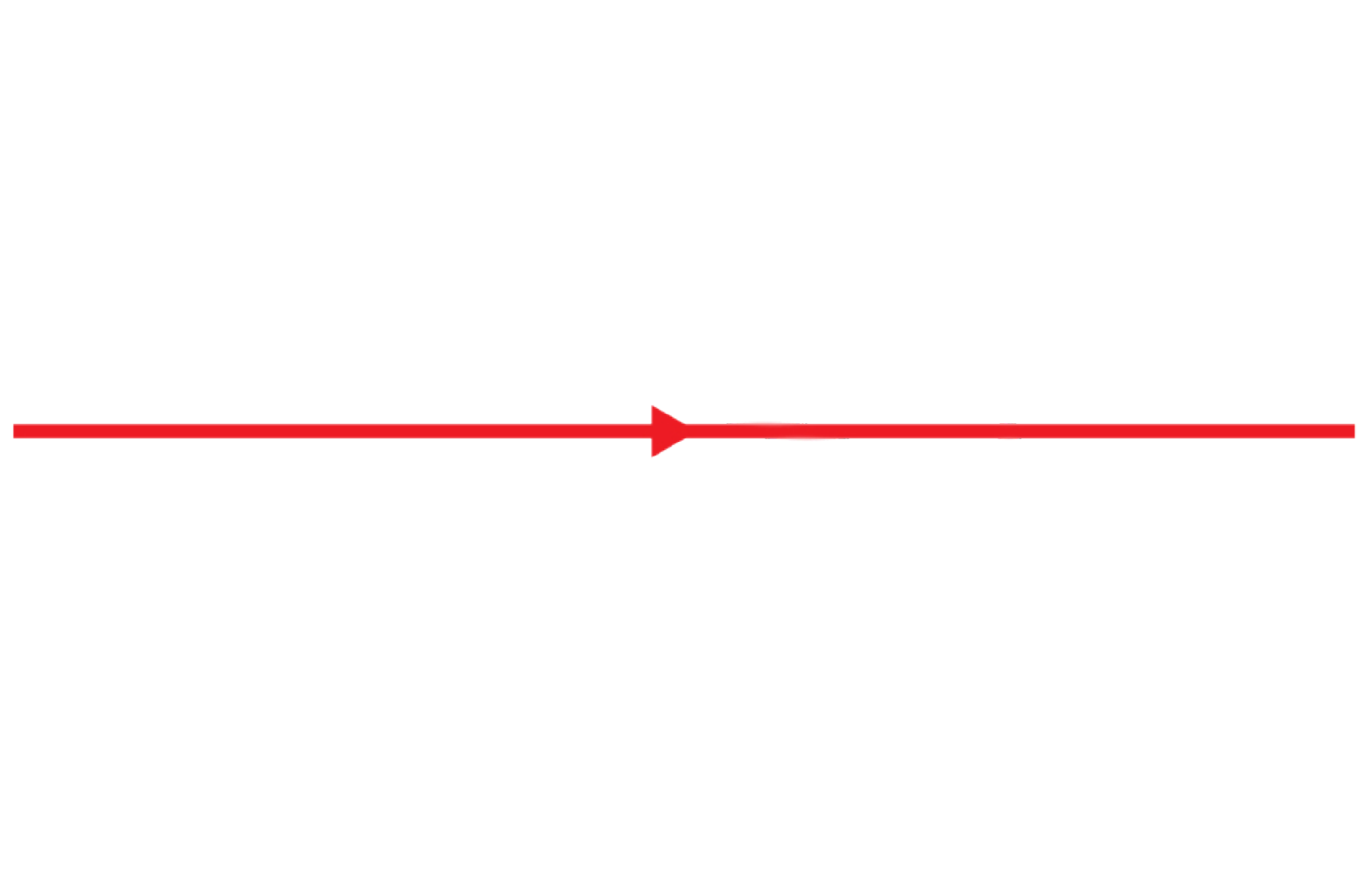}
      \put(102.5,29){\small $\Sigma$}
    \end{overpic} \qquad\qquad
     \begin{overpic}[width=.4\textwidth]{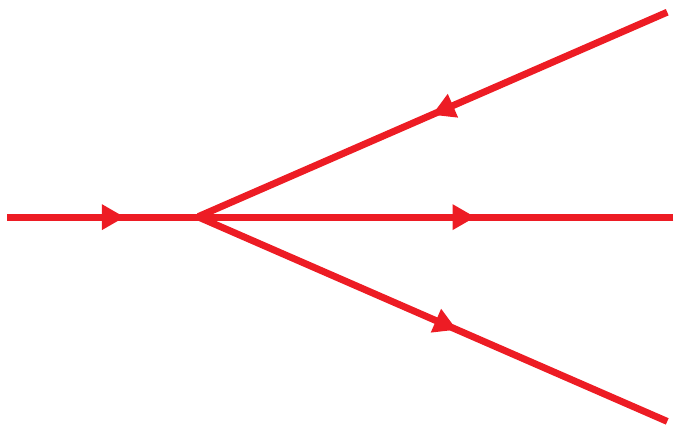}
      \put(102,29){\small $\hat{\Sigma}$}
      \put(49,44){\small $\gamma$}
    \end{overpic}
     \begin{figuretext}\label{Sigmahatdeform2.pdf}
        The contour  $\Sigma$ (left) and the contour $\hat{\Sigma} = \Sigma \cup \gamma$ (right). 
         \end{figuretext}
     \end{center}
\end{figure}

\begin{lemma}\label{deformationlemma2}
Let $\Sigma$ be a piecewise smooth contour and let $\hat{\Sigma} = \Sigma \cup \gamma$ denote $\Sigma$ with a lens through \(\infty\) added as in Figure \ref{Sigmahatdeform2.pdf}.
Consider the $2 \times 2$-matrix RH problem
\begin{align}\label{RHmdeform}
\begin{cases}
N \in I + \dot{E}^2(\C \setminus \Sigma),\\
N_+(k) = N_-(k) V_N(k) \quad \text{for a.e.} \ k \in \Sigma,
\end{cases}
\end{align}
where $V_N: \Sigma \to GL(2,\C)$ is a jump matrix defined on $\Sigma$ and let $u$ be a $2 \times 2$-matrix valued function such that $u - u(\infty) \in (\dot{E}^2 \cap E^\infty) (\C \setminus \hat{\Sigma})$ and $u^{-1} - u(\infty)^{-1} \in (\dot{E}^2 \cap E^\infty) (\C \setminus \hat{\Sigma})$. 
Then $m$ satisfies the RH problem (\ref{RHmdeform}) if and only if the function $\hat{m}$ defined by
\begin{align}\label{deformlemmadefincasinfty}
\hat{N}(k) = u(\infty)^{-1} N(k) u(k), \qquad k \in \C \setminus \hat{\Sigma},\end{align}
satisfies the RH problem 
\begin{align}\label{RHmhatdeform}
\begin{cases}
\hat{N} \in I + \dot{E}^2(\C \setminus \hat{\Sigma}),\\
\hat{N}_+(k) = \hat{N}_-(k) \hat{v}_N(k) \quad \text{for a.e.} \ k \in \hat{\Sigma},
\end{cases}
\end{align}
with
\begin{align}\label{hatvdeform}
\hat{v}_N = \begin{cases}
u_-^{-1} v_N u_+, & k \in \Sigma, \\
u_-^{-1} u_+, & k \in \gamma.
\end{cases}
\end{align}
\end{lemma}

\begin{proof}
Suppose \(N\in I+\dot{E}^2(\C\setminus\Sigma)\) satisfies the RH problem \eqref{RHmdeform} and define \(\hat N(k)\) by \eqref{deformlemmadefincasinfty} for \(k\in\C\setminus\hat\Sigma\). We write
\begin{align*}
\hat N-I&=u(\infty)^{-1}(N-I)(u-u(\infty))+u(\infty)^{-1}(u-u(\infty))\\&\quad+u(\infty)^{-1}(N-I)u(\infty).
\end{align*}
Since \(N-I\in \dot{E}^2(\hat\C\setminus\hat \Gamma)\) and \(u-u(\infty)\in (\dot{E}^2 \cap E^\infty)(\hat\C \setminus \hat{\Gamma})\), it follows that that \(\hat N\in I+\dot{E}^2(\hat\C\setminus\hat\Gamma)\).
For the jump matrix we compute
$$ \hat v=\hat m_-^{-1}\hat m_+=u_-^{-1} m_-^{-1} u(\infty)u(\infty)^{-1} m_+ u_+=u_-^{-1} v u_+$$
along \(\hat \Gamma\) (where we have put \(v=1\) on \(\gamma\)). 
Thus \(\hat N\) satisfies the RH problem determined by \((\hat \Sigma,\hat{v}_N)\).
The converse statement follows similarly.
\end{proof}

According to Lemma \ref{deformationlemma2}, in order to verify that the right hand side of identity \eqref{fulltransformation} is indeed a solution of our original RH problem  \eqref{RHm}, we have to show that the matrix 
\begin{align}\label{transformationmatrix}
m^{app}D(\zeta,k)^{\sigma_3}H_4^{-1}H_2^{-1}e^{-i(t g(k) - kx - 2k^2 t)\sigma_3}-D(\zeta,\infty)^{\sigma_3}e^{-itg_\infty\sigma_3}
\end{align} lies in \((\dot{E}^2 \cap E^\infty) (\C \setminus \hat{\Gamma})\).

The \(\dot{E}^2\) condition is a  consequence of the asymptotics at infinity. 
Note that \(m^{app}-I\), \(D(\zeta,k)^{\sigma_3}-D(\zeta,\infty)^{\sigma_3}\), \(H_4^{-1}-I\) and \(e^{-i(t g(k) - kx - 2k^2 t)\sigma_3}-e^{-itg_\infty\sigma_3}\) all lie in \(O(k^{-1})\) as \(k\to\infty\). Thus the same holds for the matrix \eqref{transformationmatrix}. From this it follows that the transformation matrix lies in \(\dot{E}^2\) (cf. Lemma 3.7 in \cite{CarlesonContoursLenells}).

The boundedness of \eqref{transformationmatrix} away from the branch points is immediate. 
At the branch points the only singular matrices are \(m^{app}\),  \(D^{\sigma_3}\) and \(H_4^{-1}\), so it suffices to show that \(m^{app}D^{\sigma_3}H_4^{-1}\) is bounded near the branch points.
We only study the case \(k\to E_2\), \(\im k>0\). The other cases follow similarly.
Let \(\im k>0\).
We compute
\begin{align*}
&m^{app}(k)D^{\sigma_3}(k)H_4^{-1}(k)
	\\
	 &=\frac{1}{2}\begin{pmatrix} (\Delta+ \Delta^{-1})D &  \Delta D^{-1} (D^2r_{2,a} e^{-2it g}-i)+\Delta^{-1}D^{-1}( D^2r_{2,a}e^{-2itg}+i) \\
i(\Delta - \Delta^{-1})D &   i\Delta D^{-1} (D^2r_{2,a} e^{-2it g}-i)-i\Delta^{-1}D^{-1}( D^2r_{2,a}e^{-2itg}+i)  \end{pmatrix}.
\end{align*}
The asymptotics of \(\Delta\) and \(D\) near \(E_2\) (cf. Lemma \ref{Dlemma})  imply that the first column of the previous matrix is bounded as \(k\to E_2\).
To study the second column we have to go back to the construction of \(r_{2,a}\). Looking at the proof of Lemma \ref{analyticapproximation r2} it follows that the singular behaviour of \(r_{2,a}=f_0+f_a\) is given by \(f_0\).
In view of \eqref{asymptoticsf}, we thus may compute the behavior of \(f_0\) near \(E_2\) using \eqref{seriesexpansionr}.
It follows that
\begin{align}\label{asymptoticsr2a}
r_{2,a}(k)=f_0(k)+O(1)=\frac{i}{2q_{2,1}}\frac{1}{\sqrt{k-E_2}}+O(1)
\end{align}
 as \(k\to E_2\), where the square root denotes the principal branch of the root with a branch cut along \((-\infty,E_2)\). 
Combined with the asymptotics of \(\Delta\) and \(D\) it follows that the first summand in each entry of the second column is bounded near \(E_2\). 
Finally, to show the boundedness of the remaining terms, it suffices to show that \( D^2(k)r_{2,a}(k)e^{-2itg(k)}+i=O(|k-E_2|^{1/2})\) as \(k\to E_2\).
Using \(g(k)=O(|k-E_2|^{1/2})\), statement \eqref{asymptoticsr2a}, as well as the precise asymptotics of \(D\) near \(E_2\) provided by Lemma \ref{Dlemma} (f), we find
\begin{align*}
D^2(k)r_{2,a}(k)e^{-2itg(k)}	
	&=\frac{i}{2q_{2,1}} \frac1{\sqrt{k-E_2}}D^2(k)+O(|k-E_2|^{1/2})
	\\
	&=\frac{\sqrt{E_2-k}}{\sqrt{k-E_2}} +O(|k-E_2|^{1/2})
	\\
	&=-i+O(|k-E_2|^{1/2})
\end{align*}
as \(k\to E_2\).

\subsection{Dressing Type Arguments}

So far we have shown that the RH problem \eqref{RHm} has a solution for each \((x,t)\in S=\{0\leq\frac{x}{t}\leq c_0\}\cap\{t\geq T\}\) and that the limit in \eqref{ulim} exists.
The next step is to show that the the function \(u\) defined by \eqref{ulim} is a solution of \eqref{nls} and satisfies \eqref{uxlim}.
Note that since the NLS equation \eqref{nls} is invariant under multiplication with a constant of absolute value \(1\), and since \(|D(0,\infty)|=1\) (cf. Lemma \ref{Dlemma} (c)), we may ignore the factor \(-D(0,\infty)^2\) in \eqref{ulim} while investigating if \eqref{ulim} is a solution of \eqref{nls}.
The remainder of the proof is however standard and will therefore be omitted.
The main step in the proof is to show that the solution \(m(x,t,k)\) of the RH problem \eqref{RHm} solves the Lax pair equations
\begin{align*}
\begin{cases}
  m_x + ik[\sigma_3, m] = Um,
  	\\ 
 m_t + 2ik^2 [\sigma_3, m] = V m,	
\end{cases}
\end{align*}
with
\begin{align*}
U = \begin{pmatrix} 0 & u \\
\bar{u} & 0 \end{pmatrix}, \qquad
V = \begin{pmatrix} -i |u|^2 & 2ku + iu_x \\
2k \bar{u} - i \bar{u}_x & i |u|^2 \end{pmatrix}, \qquad \sigma_3 = \begin{pmatrix} 1 & 0 \\ 0 & -1 \end{pmatrix},
\end{align*}
and \(u\) being given by \eqref{ulim}.
For the details we refer to the proof of Theorem 7 in \cite{mKdVLenells}. Although \cite{mKdVLenells} studies the mKdV equation, the corresponding proof can be easily adapted to our situation. 
The only difference is that in \cite{mKdVLenells} a vanishing lemma is used which is missing in our situation. However, since we already know that our  RH problem \eqref{RHm} has a solution in the appropriate sector, the arguments in \cite{mKdVLenells} still apply.

\subsection{Asymptotics of \(u\)}\label{asymptoticsu}

Taking into account the asymptotics of \(g\), \(D\), \(m^{(\infty)}\),  \(H_4\), and \(\hat m\),  the equations \eqref{ulim} and \eqref{fulltransformation} yield
\begin{align}\label{exacttermu}
u(x,t) &= -2i D(0,\infty)^2 \overset{\angle}{\lim_{k \to \infty}} (k\,m(x,t,k))_{12}
	\nonumber
	\\
	\nonumber
	&=-2ie^{2itg_\infty(\zeta)}D(0,\infty)^2 D(\zeta,\infty)^{-2}\lim_{k\to \infty}(\hat mm^{(\infty)})
	\\
	&=e^{2itg_\infty(\zeta)}e^{-\frac{1}{\pi i}\int_{k_0}^{\kappa_+} \frac{\ln(1 - |r(s)|^2)}{X(s)} ds}(\alpha-2i\hat m_{1,12}(x,t)).
\end{align}
Here we have used that \(D(0,\infty)^2 D(\zeta,\infty)^{-2}=e^{-\frac{1}{\pi i}\int_{k_0}^{\kappa_+} \frac{\ln(1 - |r(s)|^2)}{X(s)} ds}\) (cf. \eqref{Dinfinitydefinition}).
Using \eqref{defginfty}, \(\omega=-2(\alpha^2+2\beta^2)\), and \eqref{mhatasymptoticspart2}, we thus arrive at
\begin{align*}
	u(x,t)&=\alpha e^{2i\beta x +it\omega}e^{-\frac{1}{\pi i}\int_{k_0}^{\kappa_+} \frac{\ln(1 - |r(s)|^2)}{X(s)} ds}
	\\
	&\quad-2ie^{2i\beta x +it\omega}e^{-\frac{1}{\pi i}\int_{k_0}^{\kappa_+} \frac{\ln(1 - |r(s)|^2)}{X(s)} ds}\left(\frac{Y(\zeta,t,k_0)m_1^X(\zeta)Y(\zeta,t,k_0)^{-1}}{\sqrt{t}\psi(\zeta,k_0)}\right)_{12}
	\\
	&\quad+O\left(\frac{\ln t}{t}\right)
\end{align*}
as \(t\to\infty\). 
Inserting the definition of \(Y\), \(m_1^X\) and \(\psi\) (see \eqref{defY}, \eqref{definitionm1X} and \eqref{defpsi}, respectively) into the above identity, the first part of Theorem \ref{maintheorem} (d) follows.

\subsection{Asymptotics of \(u_x\)}\label{asymptoticsux}

Similar to the previous section, formula \eqref{uxlim} yields
\begin{align}\label{exacttermux}
u_x(x,t)&=-D(0,\infty)^2\overset{\angle}{\lim_{k \to \infty}}\left(4k^2m_{12}(x,t,k)+2i u(x,t) k m_{22}(x,t,k)\right)
	\nonumber
	\\
	\begin{split}
	&=2e^{2i\beta x +it\omega}e^{-\frac{1}{\pi i}\int_{k_0}^{\kappa_+} \frac{\ln(1 - |r(s)|^2)}{X(s)} ds}(i\alpha (\beta-\hat m_{1,11}(x,t)+\hat m_{1,22}(x,t))
	\\&\qquad+2\hat m_{1,12}(x,t)\hat m_{1,22}(x,t)-2\hat m_{2,12}(x,t))
	\end{split}
	\\
	&=2i\alpha\beta e^{2i\beta x +it\omega}e^{-\frac{1}{\pi i}\int_{k_0}^{\kappa_+} \frac{\ln(1 - |r(s)|^2)}{X(s)} ds}
	\nonumber
	\\
	&\quad-2e^{2i\beta x +it\omega}e^{-\frac{1}{\pi i}\int_{k_0}^{\kappa_+} \frac{\ln(1 - |r(s)|^2)}{X(s)} ds}
	(i\alpha \hat m_{1,11}(x,t)\nonumber\\&\qquad-i\alpha \hat m_{1,22}(x,t)+2\hat m_{2,12}(x,t))		
	+O\left(\frac{1}{t}\right)
	\nonumber
\end{align}
as \(t\to\infty\). 
As before, after substitution \(\hat m\) as well as the functions appearing in \(\hat m\) into the above formula (see \eqref{mhatasymptoticspart2} as well as \eqref{defY}, \eqref{definitionm1X} and \eqref{defpsi}) the desired asymptotics for \(u_x\) as stated in Theorem \ref{maintheorem} follow.
This concludes the proof of Theorem \ref{maintheorem}.

\begin{remark}[Alternative Derivation of the Asymptotics of $u_x$]\label{comparisonuxderivation}

Note that by differentiating the expression \eqref{exacttermu} with respect to \(x\), we get another expression for \(u_x\) and hence another way to calculate the asymptotics of \(u_x\).
This alternative way of calculating the asymptotics of \(u_x\) leads also to a different formula of the subleading coefficients \(u_b\).
The new formula requires the asymptotics of \(\frac{d}{dx} \hat m_{1,12}\), which can be obtained by differentiating \eqref{mhatasymptoticspart2}  with respect to \(x\). The only terms of order \(O(t^{-1/2})\) stem from the case when the \(x\)-derivative is applied to the oscillatory term \(e^{-itg(\zeta,k_0)\sigma_3}\) appearing in \(Y\).
This leads to 
\begin{align*}
u^{alt}_b(x,t)&=2e^{2i\beta x +i\omega t}e^{-\frac{1}{\pi i}\int_{k_0}^{\kappa_+} \frac{\ln(1 - |r(s)|^2)}{X(s)} ds}\Big(2\beta \frac{Y(\zeta,t,k_0)m_1^X(\zeta)Y(\zeta,t,k_0)^{-1}}{\sqrt{t}\psi(\zeta,k_0)}\\&\quad-X(k_0)\frac{Y(\zeta,t,k_0)\sigma_3m_1^X(\zeta)Y(\zeta,t,k_0)^{-1}}{\sqrt{t}\psi(\zeta,k_0)}
\\&\quad+X(k_0)\frac{Y(\zeta,t,k_0)m_1^X(\zeta)\sigma_3Y(\zeta,t,k_0)^{-1}}{\sqrt{t}\psi(\zeta,k_0)}
\Big)_{12}.
\end{align*}
It is straightforward (albeit tedious) to verify that \(u^{alt}_b=u_b\).
\end{remark}



%

\bibliographystyle{plain}
\bibliography{is}

\end{document}